%% file: main_forest.tex
\documentclass[dvipsnames]{siamonline220329}
\input{shared}

\usepackage{comment}

\newsiamthm{conjecture}{Conjecture}%
\crefname{conjecture}{Conjecture}{Conjectures}
\headers{Model Forest EnKF}{A. A. Popov and  A. Sandu}

\begin{document}

\maketitle

\renewcommand{\thefootnote}{\fnsymbol{footnote}}
\footnotetext[2]{Oden Institute for Computational Engineering and Sciences, University of Texas at Austin, Austin, TX (\email{apopov@vt.edu})}
\footnotetext[3]{Computational Science Laboratory, Department of Computer Science, Virginia Tech, Blacksburg, VA   (\email{sandu@cs.vt.edu}).}
\renewcommand{\thefootnote}{\arabic{footnote}}

\begin{abstract} 
Traditional data assimilation uses information obtained from the propagation of one physics-driven model and combines it with information derived from real-world observations in order to obtain a better estimate of the truth of some natural process.
However, in many situations multiple simulation models that describe the same physical phenomenon are available. Such models can have different sources. 
On one hand there are theory-guided models are constructed from first physical principles, while on the other there are data-driven models that are constructed from snapshots of high fidelity information.
In this work we provide a possible way to make use of this collection of models in data assimilation by generalizing the idea of model hierarchies into  model forests---collections of high fidelity and low fidelity models organized in a groping of model trees such as to capture various relationships between different models.
We generalize the multifidelity ensemble Kalman filter that previously operated on model hierarchies into the model forest ensemble Kalman filter through a generalized theory of linear control variates. 
This new filter allows for much more freedom when treading the line between accuracy and speed. 
Numerical experiments with a high fidelity quasi-geostrophic model and two of its low fidelity reduced order models validate the accuracy of our approach.
\end{abstract}



\section{Introduction}

In many situations the availability of multiple models that describe the same physical system is a valuable asset for obtaining accurate forecasts. For example the Coupled Model Intercomparison Project~\cite{eyring2016overview} used by the International Panel on Climate Change is an effort to utilize an aggregate of a wide array of climate models for the purposes of increasingly accurate predictions. It is a recognition by the climate community that a collection of models is greater than the sum of its parts.

The idea of leveraging a collection of models to improve data assimilation~\cite{reich2015probabilistic,law2015data,asch2016data} has seen an explosion of research over the last several years.  Multilevel data assimilation was first developed in the context of Monte Carlo methods~\cite{giles2008multilevel,giles2015multilevel}, wherein a hierarchy of models, through successive coarsening in the time dimension, was used to perform inference with the accuracy of the finest level coarsening with a larger and larger amount of samples from the coarser levels.
The ideas of multilevel Monte Carlo were transferred to the ensemble Kalman filter (EnKF) in a series of works developing the multilevel ensemble Kalman filter (MLEnKF)~\cite{Hoel_2016_MLEnKF, chernov2020multilevel, hoel2019multilevel, hoel2021multi, chada2020multilevel} aiming to provide more operationally viable methods.

The multifidelity ensemble Kalman filter (MFEnKF)~\cite{Popov2021a, Popov2021b, popov2022multifidelitychapter, donoghueamulti} circumvents numerical difficulties present in the MLEnKF through a robust use of linear control variate theory.
The MFEnKF also extends the idea of model coarseness to arbitrary non-linear couplings between high fidelity (fine level) and low fidelity (coarse level) model states, allowing the use of various types of reduced order models (ROMs) to form a model hierarchy.

This work further extends the EnKF ideas and brings two novel contributions. (i) First, it extends model hierarchies to model trees and model forests, covering the situation were the collection of models cannot neatly form a model hierarchy. (ii) Second, it extends the multifidelity ensemble Kalman filter to the model forest Kalman filter allowing data assimilation to make use of model forests in a rigorous way.

Given one high fidelity model and a collection of low fidelity models, it is not always possible to organize them in a strict model hierarchy. Following this observation we introduce the first key contribution of the this work (i); we generalize the idea of model hierarchies to model trees, where one model is allowed to have multiple low fidelity models on the same level below it; the low fidelity models are surrogates for the high fidelity one, but they may not have a direct relationship with each other. This results in a tree structure of models with the high fidelity model acting as the root.
We further extend model trees by leveraging the idea of model averaging~\cite{dormann2018model}. Assuming that we have a collection of model trees, each with their own high fidelity model at the root, we organize them in a ``model forest'' and 
build an averaging procedure over all the trees in the forest.

By bringing together the ideas of the MFEnKF with that of model forests, we make the second key contribution (ii) of this work; we replace the MFEnKF with the model forest ensemble Kalman filter, which also has the acronym MFEnKF as we show that the former is a special case of the latter. 

Numerical tests on the Quasi-Geostrophic equations with a quadratic reduced order model and an autoencoder-based surrogate show that our proposed extension significantly decreases the number of high fidelity model runs required to achieve a certain level of analysis accuracy.

This paper is organized as follows. Relevant background information including the sequential data-assimilation problem, model hierarchies, model averages, and the multifidelity ensemble Kalman filter are presented in~\Cref{sec:background}. The extension of model hierarchies to model trees, and the extension of model averages to model forests is described in~\Cref{sec:model-forests}. Next the extension of the multifideity ensemble Kalman filter to the model forest Kalman filter is explained in~\Cref{sec:model-forest-enkf}. The quasi-geostrophic equations and two surrogate models are detailed in~\Cref{sec:models}. Numerical experiments on various model trees and model forests are presented in~\Cref{sec:numerical-experiments}. Finally, some closing remarks are stated in~\Cref{sec:conclusions}.

\section{Background}
\label{sec:background}
We review relevant background on data assimilation, including model hierarchies, linear control variates, model averaging, and the multifidelity ensemble Kalman filter.

\subsection{Data Assimilation}

Let $X_i^{\|t}$ denote the state of some natural process at time $t_i$, where the superscript $\|t$ represents ground-truth. Assume that we have some prior information about this state represented by the distribution of the random variable $X^{\|b}_i$. 
Assume also that we have access to some sparse noisy observations of the truth represented by,
\begin{equation}\label{eq:observations}
    Y_i = \!H(X_i^{\|t}) + \varepsilon_i,
\end{equation}
where $\!H$ is a non-linear observation operator and $\epsilon_i$ is a random variable representing observation error. For the remainder of this paper we assume that the observation error is normal with distribution
\begin{equation}
    \varepsilon_i \sim \!N(\*0, \Cov{Y_i}{Y_i}).
\end{equation}

Finally, assume we have some inexact numerical model $\Model{}$ that approximates the dynamics of the natural process, i.e., evolution of the truth,
\begin{equation}\label{eq:model}
    X_i^{\|t} = \Model{}(X_{i-1}^{\|t}) + \xi_i,
\end{equation}
where the random variable $\xi_i$ represents the model error. 

Data assimilation~\cite{reich2015probabilistic,asch2016data,evensen2022data} seeks to combine the prior information $X^{\|b}_i$ with the sparse noisy observations $Y_i$ into a posterior representation $X^{\|a}_i$ of the information, commonly through Bayesian inference, 
\begin{equation}
    \pi(X^{\|a}_i) = \pi(X^{\|b}_i \mid Y_i)\propto \pi(Y_i \mid X^{\|b}_i)\,\pi(X^{\|b}_i),
\end{equation}
where the distribution $\pi(X^{\|a}_i)$ represents our full knowledge about the state of the system at time $t_i$.

The model~\cref{eq:model} also forecasts the posterior information at time index $i-i$ to prior information at time $i$, through the relation,
\begin{equation}
    X_i^{\|b} = \Model{}(X_{i-1}^{\|a}) + \xi_i.
\end{equation}

\subsection{Notation}
In this work, the mean of the random variance $X$ is denoted by, 
$\Mean{X}$,
and the covariance between the random variable $X$ and the random variable $Y$ is denoted by,
$\Cov{X}{Y}$.
An ensemble of $N$ samples from the random variable $X$ is denoted by,
    $\En{X} = \left[\*X_1, \*X_2, \dots, \*X_N\right]$,
with the ensemble mean denoted by,
\begin{equation*}
    \MeanE{X} = \sum_{i=1}^N \frac{1}{N} X_i,
\end{equation*}
the scaled ensemble anomalies denoted by,
\begin{equation*}
    \An{X} = \frac{1}{\sqrt{N-1}}\left(\En{X} - \MeanE{X}\,\*1_N^T\right),
\end{equation*}
where $\*1_N$ is a column vector of $N$ ones, and the unbiased sample covariance between $X$ and $Y$ denoted by,
$\CovE{X}{Y} = \An{X}\An{Y}^T$.

\subsection{Model Hierarchies and Order Reduction}
\label{sec:model-hierarchies}

\begin{figure}
    \centering
    \includegraphics[]{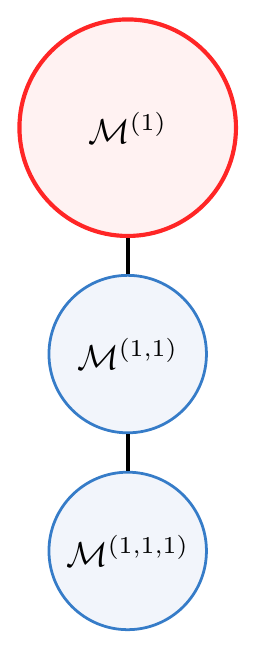}
    \caption[A visual representation of a model hierarchy with two surrogate models.]{A visual representation of a model hierarchy with two surrogate models. The principal model \mainmodel, $\Model{(1)}$, has a surrogate model \surrmodel, $\Model{(1,1)}$, which in turn has its own surrogate model \surrmodel, $\Model{(1,1,1)}$.}
    \label{fig:model-hierarchy}
\end{figure}

Assume there exists a model which is expensive to compute from which we are attempting to glean some information through a sampling procedure. Call this model the principal model. Assume that there exists a surrogate model with which we can bootstrap our knowledge about the principal model. We can then assume that the previously mentioned surrogate model is its own principal model in its own model hierarchy that has its own surrogate model. This process can be repeated \textit{ad infinitum} to obtain a model hierarchy of a desired size. \Cref{fig:model-hierarchy} provides an illustration of a model hierarchy for one principal model which has a surrogate that itself has a surrogate.

Let the tuple $\!I$ represent the index of a model in the model hierarchy, such that the model $\Model{\!I}$ has a surrogate model $\Model{\!I\cdot 1}$, with `$\,\cdot\,$' representing tuple concatenation, e.g., $(1,2)\cdot 3 = (1,2,3)$. This particular notation helps with defining model trees and model forests later.

We make the following assumptions:
\begin{itemize}
\item The dynamics of the high fidelity `principal' model $\Model{\!I}$ is embedded into the space $\mathbb{X}^{\!I}$,  i.e., $\Model{\!I}: \mathbb{X}^{\!I} \to \mathbb{X}^{\!I}$.
\item The dynamics of the low fidelity `surrogate' model $\Model{\!I\cdot 1}$ is embedded into the reduced space $\mathbb{X}^{\!I\cdot 1}$, i.e.,  $\Model{\!I\cdot 1} : \mathbb{X}^{\!I\cdot 1} \to \mathbb{X}^{\!I\cdot 1}$.
\item There exists a (possibly non-linear) projection operator that maps the states of the principal model to its surrogate:
\begin{equation}\label{eq:projection}
    \theta^{\!I\cdot 1} : \mathbb{X}^{\!I} \rightarrow \mathbb{X}^{\!I\cdot 1}.
\end{equation}
\item There exists an interpolation operator that reconstructs an approximation of the state of the principal model from that of the surrogate model:
\begin{equation}\label{eq:interpolation}
    \phi^{\!I\cdot 1} : \mathbb{X}^{\!I\cdot 1} \rightarrow \mathbb{X}^{\!I},
\end{equation}
\item The two operators obey the right-invertible consistency property~\cite{Popov2021b},
\begin{equation}
     \theta^{\!I\cdot 1} \circ  \phi^{\!I\cdot 1} =  \text{id},
\end{equation}
ensuring that reconstruction has the same representation of the full order information in the reduced space.
\end{itemize}

\subsection{Linear Control Variates for Model Hierarchies}\label{sec:linear-cv-model-hier}

We discuss the specific case of a bifidelity  model hierarchy, \treeonesurr, with the high fidelity model having one surrogate. Assume that the information about our high fidelity model run is represented by the distribution of the random variable $X$ known as the principal variate. Assume also that there exist two random variables whose distributions describe the information about the surrogate model: the control variate $\widehat{U}$ which is highly correlated to $X$, and the ancillary variate $U$ which is uncorrelated with the other variates, but shares its mean with $\widehat{U}$. The variates $X$, $\widehat{U}$ and $U$ are known as the constituent variates.

Given some (possibly non-linear) functions $h$ and $g$, the total variate which describes the total information of the hierarchy \treeonesurr{} in the linear control variate framework is given by,
\begin{equation}\label{eq:control-variates}
    Z^h = h(X) - \*S\left[ g(\widehat{U}) - g(U) \right]
\end{equation} 
where $\*S$ is known as the gain operator. The choice of $h$ and $g$ largely depends on, and defines, the information that is encapsulated by the different variates, and has to be carefully chosen for each given problem.

\begin{theorem}\label{thm:optimal-gain}
    The optimal gain matrix $\*S$ that minimizes the trace generalized variance of $Z$ in \eqref{eq:control-variates} is given by,
    \begin{equation}\label{eq:optimal-gain}
        \*S = \Cov{h(X)}{g(\widehat{U})}\left(\Cov{g(\widehat{U})}{g(\widehat{U})} + \Cov{g(U)}{g(U)}\right)^{-1}.
    \end{equation}
\end{theorem}
\begin{proof}
    By \cite{petersen2008matrix}, the derivative with respect to $\*S$ of the trace generalized variance of $Z$ is
    \begin{equation*}
        \frac{\partial}{\partial \*S} \tr(\Cov{Z}{Z}) = -2 \Cov{h(X)}{g(\widehat{U})} + 2\*S\left(\Cov{g(\widehat{U})}{g(\widehat{U})} + \Cov{g(U)}{g(U)}\right),
    \end{equation*}
    and as the Hessian is always symmetric positive definite,
    \begin{equation*}
        \frac{\partial^2}{\partial \*S^2} \tr(\Cov{Z}{Z}) = 2\left(\Cov{g(\widehat{U})}{g(\widehat{U})} + \Cov{g(U)}{g(U)}\right)\otimes\*I \geq 0,
    \end{equation*}
    the global minimum is attained when,
    \begin{equation}
        \frac{\partial}{\partial \*S} \tr(\Cov{Z}{Z}) = \*0,
    \end{equation}
    which is satisfied by \eqref{eq:optimal-gain}, as required.
\end{proof}

We now describe the generalization to a model hierarchy. 
Assume that the ancillary variate with indexing tuple $\!I$ is the total variate estimator for 
\begin{equation}\label{eq:model-hierarchy-cv}
    U^{\!I} = h^{\!I}\left(X^{\!I}\right) - \*S^{\!I \cdot 1}\left[ g^{\!I\cdot 1}\big(\widehat{U}^{\!I \cdot 1}\big) - g^{\!I\cdot 1}\big(U^{\!I \cdot 1}\big)\right],
\end{equation}
with $Z^h \coloneqq U^{(1)}$ representing the total variate

We now make the assumption that all $h^{\!I} \coloneqq \text{id}$, the identity function, that the control variate is a projection of the principal variate,
\begin{equation}
    \widehat{U}^{\!I \cdot 1} = \theta^{\!I \cdot 1}(X^{\!I}),
\end{equation}
and that the non-linear function $g^{\!I\cdot 1}$ is the interpolation operator,
\begin{equation}
    g^{\!I\cdot 1} = \phi^{\!I\cdot 1}.
\end{equation}

The following is a known result from \cite{Popov2021a},  which holds for non-linear operators,
\begin{theorem}\label{thm:popov-gain}
    Under the assumption that the reconstruction error is uncorrelated from the original state,
    \begin{equation*}
        \Cov{X^{\!I}}{\left(X^{\!I} - \phi^{\!I\cdot 1}(\widehat{U}^{\!I \cdot 1})\right)} = \*0,
    \end{equation*}
    and the assumption that the reconstruction covariances for the control and ancillary variates are identical, 
    \begin{equation*}
        \Cov{\phi^{\!I\cdot 1}(\widehat{U}^{\!I\cdot 1})}{\phi^{\!I\cdot 1}(\widehat{U}^{\!I\cdot 1})} = \Cov{\phi^{\!I\cdot 1}(U^{\!I\cdot 1})}{\phi^{\!I\cdot 1}(U^{\!I\cdot 1})},
    \end{equation*}
    the optimal gain is,
    \begin{equation*}
        \*S^{\!I\cdot 1} = \frac{1}{2}\,\*I.
    \end{equation*}
\end{theorem}

\begin{remark}
    The assumption that the reconstruction error is uncorrelated from the original state is, unfounded. We hypothesize that it would be beneficial for this framework if projection and interpolation operators were built with it in mind, and attempted to minimize the error involved.
\end{remark}

\subsubsection{Forecasting with Model Hierarchies}
\label{sec:model-propagation-for-model-hier}

We now describe how model propagation is handled for the linear control variate setting.
As each model hierarchy can be written in a nested manner, we can, without loss of generality discuss the propagation of the bifidelity tree \treeonesurr{}.

Assume that we have a high fidelity model, $\Model{(1)}$ that acts on the principal variate $X_i$ and a low fidelity model $\Model{(1,1)}$ that acts on the control variate $\widehat{U}_i$ and ancillary variate $U_i$, all of which are propagated from time $i$ to time $i+1$.
We assume the following natural decomposition of the model propagation,
\begin{equation}\label{eq:linear-control-variate-model-propagation}
    \Model{(1)}(Z_i) \underset{\text{assumed}}{=} \Model{(1)}(X_i) - \*S^{(1,1)}\left[\phi^{(1)}\left(\Model{(1,1)}(\widehat{U}_i)\right) - \phi^{(1)}\left(\Model{(1,1)}(U_i)\right)\right],
\end{equation}
meaning that the action of the high fidelity model on the total variate $Z_i$, at time $i$, is defined in terms of the linear control variate framework. The models act on the constituent variates at their respective fidelities, thus the model~\cref{eq:linear-control-variate-model-propagation} is an approximate evolution model for $Z$.

\begin{remark}
    The assumption made in~\cref{eq:linear-control-variate-model-propagation} is not true in the general case. Accounting for the model error generated by this assumption is of independent interest.
\end{remark}

\subsection{Model Averaging}
\label{sec:model-avaraging}

\begin{figure}[t]
    \centering
    \includegraphics[]{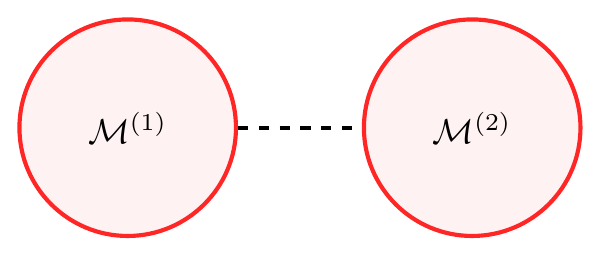}
    \caption[Model averaging]{An visual representation of model averaging, where two models,  $\Model{(1)}$ and $\Model{(2)}$ are both high fidelity \mainmodel{} models.}
    \label{fig:model-average}
\end{figure}

In this work, the term model averaging refers to the idea of combining, in a Bayesian sense, information from multiple independent high fidelity models. \cref{fig:model-average} shows a visual representation of a model average.

We base our interpretation of model averaging on previous work on Bayesian model averaging techniques~\cite{dormann2018model} and on multimodel ensemble Kalman filtering~\cite{xue2014multimodel}, but  modify the exposition for our purposes.
Of note is a new, related concept of supermodeling in data assimilation \cite{sendera2020supermodeling}, which is outside the scope of this paper.

Assume that we have a collection of random variables $\{X^{(m)}\}_{m=1}^M$ each of which representing some (potentially biased) information about the truth in a collection of spaces $\{\mathbb{X}^{(m)}\}_{m=1}^M$.
Our aim  is to combine these random variables in some optimal way. We do not know which random variable is more representative, but we can assume that we have some information about our confidence in each source of information, represented by the non-negative weights $\{w^{(m)}\}_{m=1}^M$ that sum up to one:
\begin{equation*}
    w^{(m)} > 0, \qquad \sum_{m=1}^M w^{(m)} = 1,
\end{equation*}
One can think of each weight $w^{(m)}$ as representing the probability of choosing the corresponding model $\Model{(m)}$ in a uni-fidelity setting.

Assume that we have transfer operators,
\begin{equation}
    \tau^{(m,m')} : \mathbb{X}^{(m')} \to \mathbb{X}^{(m)}
\end{equation}
that transfer the state from the space of model $m'$ to the space of model $m$, similar to the projection~\cref{eq:projection} and interpolation~\cref{eq:interpolation} operators, with $\tau^{(m,m)} \coloneqq \text{id}$ for all $m$.
The weighted average of all our information defined in the space $\mathbb{X}^{(m)}$ can be denoted by,
\begin{equation}\label{eq:model-averaging}
    \!X^{(m)} = \sum_{m'=1}^M w^{(m')}\,\,\tau^{(m, m')} (X^{(m')}), 
\end{equation}
with the resulting variable $\!X^{(m)}$ residing in the space $\mathbb{X}^{(m)}$ of model $\Model{(m)}$.

As we are dealing with Kalman filter family methods, the first two moments of $\!X^{(m)}$ are of particular interest, and are given by,
\begin{gather}
    \Mean{\!X^{(m)}} = \sum_{m'=1}^M w^{(m')}\Mean{\tau^{(m,m')}(X^{(m')})},\\ \Cov{\!X^{(m)}}{\!X^{(m')}} = \sum_{m''=1}^M \sum_{m'''=1}^M w^{(m'')} w^{(m''')} \Cov{\tau^{(m,m'')}(X^{(m'')})}{\tau^{(m',m''')}(X^{(m''')})}.
\end{gather}

\subsubsection{Forecasting with Model Averages}
\label{sec:model-propagation-for-model-averages}
We now show how model averages can be  propagated in time.

Specifically for the two-model case \modelaverage{} we have the random variables $X^{(1)}_i$ and $X^{(2)}_i$ at time index $i$. We can build the model average,
\begin{equation}
   \!X_i^{(1)} = w^{(1)} X^{(1)}_i + w^{(2)} \tau^{(1,2)}(X^{(2)}_i),
\end{equation}
in the space $\mathbb{X}^{(1)}$ where the models have the weights $w^{(1)}$ and $w^{(2)}$.

If the corresponding models for each random variable are $\Model{(1)}$ and $\Model{(2)}$, then, 
similar to the model propagation of model hierarchies~\cref{eq:linear-control-variate-model-propagation}, we consider the following propagation of the random averages:
\begin{equation}\label{eq:model-average-propagation}
   \underbrace{  \Model{}(\!X_i^{(1)})  }_{\!X^{(1)}_{i+1}} \underset{\text{assumed}}{=}  w^{(1)} \underbrace{ \Model{(1)}(X^{(1)}_i) }_{X^{(1)}_{i+1}} + w^{(2)} \tau^{(1,2)}\big(\underbrace{\Model{(2)}(X^{(2)}_i) }_{X^{(2)}_{i+1}}\big),
\end{equation}
with the model $\Model{}$ representing an implicit virtual model combining the propagation of the constituent models.

Practical aspects need to be considered when estimating the cross covariance term,
\begin{equation}
    \Cov{\Model{(1)}(X^{(1)}_i)}{\tau^{(1,2)}\left(\Model{(2)}(X^{(2)}_i)\right)},
\end{equation}
and other similar terms, as pairwise sampling is required to build valid sample covariances.
One approach to overcome this issue is to approximate the covariance via the following approach,
\begin{multline}\label{eq:covariance-ensemble-average}
   \Cov{\Model{(1)}(X^{(1)}_i)}{\tau^{(1,2)}\left(\Model{(2)}(X^{(2)}_i)\right)} \approx\\
    \frac{1}{2}\left(\Cov{\Model{(1)}(X^{(1)}_i)}{\tau^{(1,2)}\left(\Model{(2)}(\tau^{(2,1)}(X^{(1)}_i))\right)} + \Cov{\Model{(1)}(\tau^{(1,2)}(X^{(2)}_i))}{\tau^{(1,2)}\left(\Model{(2)}(X^{(2)}_i)\right)}\right),
\end{multline}
where, as the distribution of $X^{(1)}$ approaches the distribution of $X^{(2)}$, so too does the estimate in \cref{eq:covariance-ensemble-average}. Which we can formalize as follows.

\begin{lemma}
Under identity assumptions on $\tau^{(1,2)}$ and  $\tau^{(2,1)}$, as the random variables of the states approach each other in distribution, $X^{(2)} \xrightarrow[]{d} X^{(1)}$ the error in \cref{eq:covariance-ensemble-average} goes to zero.
\end{lemma}

Intuitively, if the random variables $X^{(1)}$ and $X^{(2)}$ are not strongly correlated, then the cross covariance  estimate in \cref{eq:covariance-ensemble-average} is an overestimate of the `true' uncertainty. An overestimate is better than no estimate.

\subsection{Multifidelity Ensemble Kalman Filter}
\label{sec:multifilidety-enkf}
In this section we give a short overview of the perturbed observations ensemble Kalman filter and the perturbed observations multifidelity ensemble Kalman filter.

Given a unifidelity model, \mainmodel, acting on variable $X$, 
an ensemble of $N$ samples $\En{X^\|b_i}$, from the prior distribution, and an observation of the truth $Y_i$, both at time index $i$, the perturbed observations ensemble Kalman filter,
\begin{equation}
    \En{X^\|a_i} = \En{X^\|b_i} - \*K_i\left(\!H_i(\En{X^\|b_i}) - \En{Y_i}\right),
\end{equation}
produces and ensemble $\En{X^\|a_i}$ of samples from an approximate posterior distribution, where $\*K_i$ is the sample Kalman gain, defined as,
\begin{equation}
    \*K_i = \CovE{X^\|b_i}{\!H_i(X^\|b_i)}{\left(\CovE{\!H_i(X^\|b_i)}{\!H_i(X^\|b_i)} + \Cov{Y_i}{Y_i}\right)}^{-1},
\end{equation}
and $\En{Y_i}$ is an ensemble of perturbed observations. For a more detailed look at this formulation of the EnKF see~\cite{Burgers_1998_EnKF}.

Given a bifidelity model hierarchy \treeonesurr{}, the multifidelity ensemble Kalman filter operates on three ensembles: the ensemble $\En{X^\|b_i}$ of the principal variate representing the dynamics of the high fidelity model, the ensemble $\En{\widehat{U}^\|b_i}$ of the control variate representing the dynamics of the low fidelity model applied to the high fidelity samples, and the ensemble $\En{U^\|b_i}$ operating on the low fidelity model.
Instead of attempting to find an ensemble for the total variate $Z^\|b_i$, the three constituent ensembles
\begin{equation}\label{eq:constituent-ensembles}
    \left(\En{X^\|b_i}, \En{\widehat{U}^\|b_i}, \En{U^\|b_i}\right),
\end{equation}
are operated on by the 
the multifidelity ensemble Kalman filter as follows,
\begin{equation}\label{eq:multifidelity-Kalman-filter}
\begin{aligned}
    \En{X^\|a_i} &= \En{X^\|b_i} - \*K^{(1)}_i\left[\!H^{(1)}_i\left(\En{X^\|b_i}\right) - \En{Y_i}\right],\\
    \En{\widehat{U}^\|a_i} &= \En{\widehat{U}^\|b_i} - \*K^{(1,1)}_i\left[\!H^{(1,1)}_i\left(\En{\widehat{U}^{\|b}_i}\right) - \En{Y_i}\right],\\  
    \En{U^\|a_i} &= \En{U^\|b_i} - \*K^{(1,1)}_i\left[\!H^{(1,1)}_i\left(\En{U^{\|b}_i}\right) - \En{Y_i}\right],
\end{aligned}
\end{equation}
where there are a few extra terms to define. There are now two observations operators, the first is the original high fidelity observation operator from~\cref{eq:observations}, 
\begin{equation*}
    \!H^{(1)}_i \coloneqq \!H_i,
\end{equation*}
with the other observation operator bridging the gap between the low fidelity space and observation space. In this work we define it in the most natural terms in terms of the projection operator and the high fidelity observation operator, though this is not necessarily optimal,
\begin{equation}\label{eq:low fidelity-observation-operator}
    \!H^{(1,1)}_i \coloneqq \!H_i \circ \phi^{(1,1)}_i.
\end{equation}
The statistical covariance of the total variate $Z_i$ in the full space, can be approximated by,
\begin{multline}
    \CovE{Z_i}{Z_i} \approx \CovE{X_i}{X_i} - \frac{1}{2}\CovE{X_i}{\phi^{(1,1)_i}(\widehat{U}_i)}  - \frac{1}{2}\CovE{\phi^{(1,1)_i}(\widehat{U}_i)}{X_i}\\+ \frac{1}{4}\CovE{\phi^{(1,1)_i}(\widehat{U}_i)}{\phi^{(1,1)_i}(\widehat{U}_i)} + \frac{1}{4}\CovE{\phi^{(1,1)_i}(U_i)}{\phi^{(1,1)_i}(U_i)},
\end{multline}
and the covariance of the total variate in the space of the low fidelity model can be approximated by,
\begin{multline}
    \CovE{\theta^{(1,1)_i}(Z_i)}{\theta^{(1,1)_i}(Z_i)} \approx \CovE{\theta^{(1,1)_i}(X_i)}{\theta^{(1,1)_i}(X_i)}\\ - \frac{1}{2}\CovE{\theta^{(1,1)_i}(X_i)}{\widehat{U}_i} - \frac{1}{2}\CovE{\widehat{U}_i}{\theta^{(1,1)_i}(X_i)}
     + \frac{1}{4}\CovE{\widehat{U}_i}{\widehat{U}_i}
     + \frac{1}{4}\CovE{U_i}{U_i},
\end{multline}
with the two Kalman gains defined as,
\begin{equation}\label{eq:fidelity-Kalman-gains}
\begin{gathered}
    \*K^{(1)}_i = \CovE{Z_i}{\!H^{(1)}_i(Z_i)}{\left(\CovE{\!H^{(1)}_i(Z_i)}{\!H^{(1)}_i(Z_i)} + \Cov{Y_i}{Y_i}\right)}^{-1},\\
    \begin{aligned} \*K^{(1,1)}_i =\,\, & 
    \CovE{\theta^{(1,1)}_i(Z_i)_i}{\!H^{(1,1)}_i(\theta^{(1,1)}_i(Z_i))}\\&\cdot{\left(\CovE{\!H^{(1,1)}_i(\theta^{(1,1)}_i(Z_i))}{\!H^{(1,1)}_i(\theta^{(1,1)}_i(Z_i))} + \Cov{Y_i}{Y_i}\right)}^{-1}. \end{aligned}
\end{gathered}
\end{equation}
More details about the multifidelity ensemble Kalman filter can be found in~\cite{Popov2021a,Popov2021b}.

\begin{remark}[Natural decomposition]
The multifidelity Kalman filter formulas in~\cref{eq:multifidelity-Kalman-filter} is derived from a natural decomposition of the total variate~\cref{eq:control-variates}, however other decompositions are possible, and are of independent interest.
\end{remark}

\begin{remark}[Inflation]\label{rem:inflation}
The ensemble Kalman filter requires covariance inflation in order for the method to converge in finite time with a finite ensemble~\cite{popov2020explicit}. This is likely true for the multifidelity ensemble Kalman filter as well, thus after every forecast step the anomalies of the principal and control variate are scaled by some inflation factor $\alpha_X > 1$, and the anomalies of the ancillary variate are scaled by some inflation factor $\alpha_U > 1$.
\end{remark}

\begin{remark}[MFEnKF heuristics]\label{rem:MFEnKF-heuristics}
In order for the linear control variate assumptions in~\cref{sec:linear-cv-model-hier} to remain valid, an imporatant heuristic is the correction of the mean of each of the constituent ensembles to match the mean of the total variate,
\begin{align*}
    \MeanE{X^\|a_i} &\xleftarrow[]{} \MeanE{Z^\|a_i},\\
    \MeanE{\widehat{U}^\|a_i} &\xleftarrow[]{} \MeanE{\theta^{(1,1)}(Z^\|a_i)},\\
    \MeanE{U^\|a_i} &\xleftarrow[]{} \MeanE{\theta^{(1,1)}(Z^\|a_i)},
\end{align*}
as this has significantly increased the accuracy in the MFEnKF in the authors' experience.
Another vital heuristic is the re-initialization of the analysis control variate ensemble from the ensemble of the analysis principal variate through the projection operator~\cref{eq:projection},
\begin{equation*}
    \En{\widehat{U}^\|a_i} \xleftarrow[]{} \theta^{(1,1)}(\En{X^\|a_i})
\end{equation*}
ensuring that the two ensembles do not  become too decorrelated through model propagation.
\end{remark}

\section{Model Trees and Model Forests}
\label{sec:model-forests}
\begin{figure}[t!]
    \centering
    \begin{tikzpicture}
        \node[MainModel] (t1root) at (0,0) {$\Model{(1)}$};
        
        \node[SurrModel] (t1s1) at (1.6, -1.6) {$\Model{(1,2)}$};
        \node[SurrModel] (t1s2) at (1.6,-3.6) {$\Model{(1,2,1)}$};
        \node[SurrModel] (t1s3) at (-1.6,-1.6) {$\Model{(1,1)}$};
        \draw[very thick] (t1root) -- (t1s1);
        \draw[very thick] (t1s1) -- (t1s2);
        \draw[very thick] (t1root) -- (t1s3);
    \end{tikzpicture}%
    \caption[Model Tree]{An example showing a model tree where the high fidelity model \mainmodel{}, $\Model{(1)}$ that has two low fidelity surrogates \surrmodel{}, $\Model{(1,1)}$ and $\Model{(1,2)}$ one of which has its own surrogate \surrmodel{}, $\Model{(1,2,1)}$, that are not all in a hierarchy and thus form a tree-like structure.}
    \label{fig:tree-example}
\end{figure}
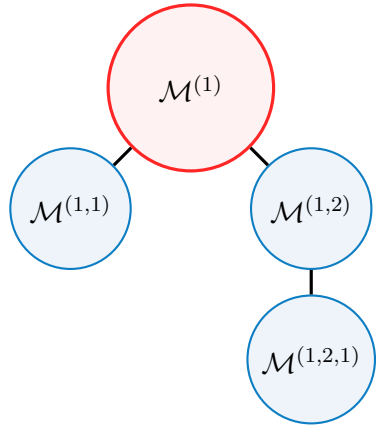

\begin{figure}[t!]
    \centering
    \begin{tikzpicture}
        \node[MainModel] (t1root) at (1,1) {$\Model{(1)}$};
        \node[MainModel] (t2root) at (5,2) {$\Model{(2)}$};
        \node[MainModel] (t3root) at (2,5) {$\Model{(3)}$};
        \draw[dashed, very thick] (t1root) -- (t2root);
        \draw[dashed,very thick] (t2root) -- (t3root);
        \draw[dashed,very thick] (t3root) -- (t1root);
        
        \node[SurrModel] (t1s1) at (-1.5,0.5) {$\Model{(1,1)}$};
        \node[SurrModel] (t1s2) at (0.5,-1.5) {$\Model{(1,2)}$};
        \node[SurrModel] (t1s21) at (-1.1,-3.1) {$\Model{(1,2,1)}$};
        \draw[very thick] (t1root) -- (t1s1);
        \draw[very thick] (t1root) -- (t1s2);
        \draw[very thick] (t1s2) -- (t1s21);

        \node[SurrModel] (t2s1) at (7.2,0.7) {$\Model{(2,1)}$};
        \draw[very thick] (t2root) -- (t2s1);
    \end{tikzpicture}
    \caption[Model Forest]{An example of a model forest, consisting of three model trees. The first tree is the same tree as described by~\cref{fig:tree-example}, the second is a simple bifidelity  tree \treeonesurr{}, consisting of the high fidelity \mainmodel{} model $\Model{(2)}$, and the low fidelity \surrmodel{} surrogate $\Model{(2,1)}$, while the third tree consists solely of the high fidelity \mainmodel{} model $\Model{(3)}$.}
    \label{fig:model-forest}
\end{figure}
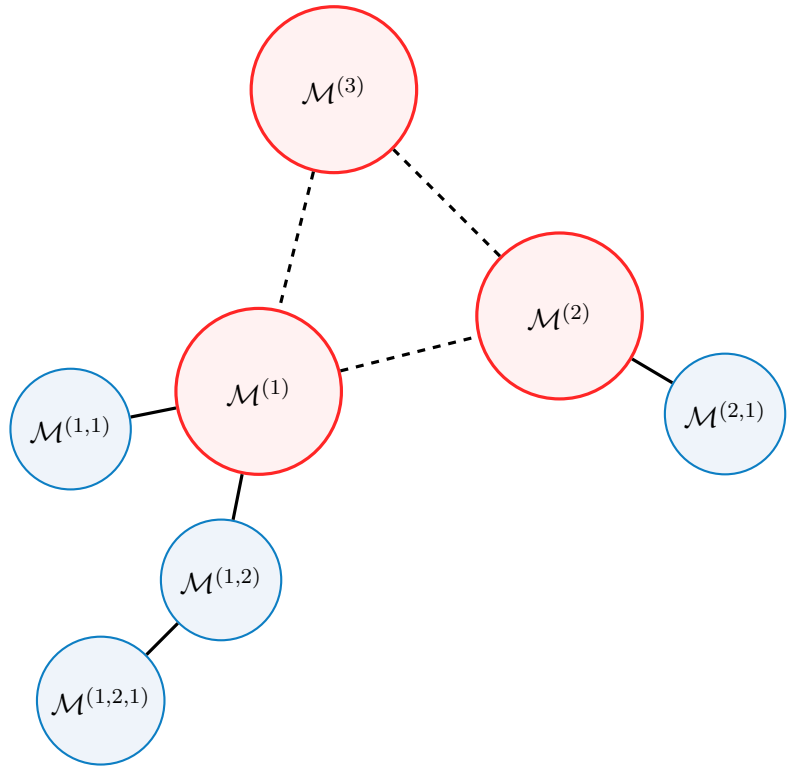

Given a collection of surrogate models it is not always possible to arrange them into a single model hierarchy. For instance, assume our high fidelity model $\Model{(1)}$ is some fine mesh discretization of a PDE. Assume also that we have two surrogates of this model: one surrogate $\Model{(1,1)}$ based on a coarsening of the mesh, and one surrogate $\Model{(1,2)}$ based on simplifying physics on the fine mesh. Both are  less accurate models than the high fidelity model, but they lose accuracy in different ways. We cannot organize them into a single `linear' model hierarchy without one simplification dominating over the other.

What we can do instead is generalize the model hierarchy idea presented in \cref{sec:model-hierarchies} to what we call ``model trees''.
Instead of each model in the hierarchy being able to have one surrogate, we instead focus on the case where each model can have multiple surrogates, thus begetting a tree structure, as can be seen contrasting figure \cref{fig:model-hierarchy} depicting a model hierarchy to \cref{fig:tree-example} depicting a model tree. In each model tree, the high fidelity model is represented by the root node \mainmodel{}, and each surrogate model is represented by the nodes \surrmodel{}.

Following the model averaging technique introduced in~\cref{sec:model-avaraging} we can build averages over collections of model trees. We call this technique ``model forests''. \Cref{fig:model-forest} provides a detailed example of a model forest.

\begin{remark}[Constructing Model Forests]
In this work, we do not explore how a model tree or model forest can be constructed. The techniques presented by multifidelity networks~\cite{mfnets} can potentially be utilized to build model trees, and extended to build model forests, though this is significantly outside the scope of this paper.
\end{remark}

\begin{remark}[Model Forests of Model Forests]
It is possible for every model in a model forest itself to be a model forest, and each model in such a model forest itself to also be a model forest, etc. This potentially cumbersome extension would necessitate fully automated ways of creating and using forests, which is significantly outside the scope of this work.
\end{remark}

\subsection{Linear Control Variates for Model Forests}

We now describe the generalization of the linear control variate framework \cref{eq:model-hierarchy-cv} from model hierarchies  first to model trees and then to model forests. 

We provide a recursive relation which defines the control variate structure of each model tree. Consider a node $\!I$, its children $\!I\cdot m$ for $m=1,\dots,M^{\!I}$, and (if $\!I$ is not the root) its parent  $\hat{\!I}$ with $\!I = \hat{\!I} \cdot k$.
\begin{itemize}
\item The node $\!I$ corresponds to the model $\Model{\!I}$.  This model state $X^{\!I}$ is the principal variate of the node.
\item The total variate at node $\!I$ -- corresponding to the subtree rooted at the node -- is denoted by $Z^{\!I}$. There are $M^{\!I}$ control-ancillary variate pairs $(\widehat{U}^{\!I\cdot m}, U^{\!I\cdot m})$, each corresponding to a child node $m = 1,\dots M^{\!I}$ and to the model $\Model{\!I\cdot m}$ . The linear control variate framework gives the following total variate:
\begin{equation}\label{eq:generalized-cv}
    Z^{\!I} = h\left(X^{\!I}\right) - \sum_{m=1}^{M^{\!I}} \*S^{\!I\cdot m} \left[g^{\!I\cdot m}\big(\widehat{U}^{\!I\cdot m}\big) - g^{\!I\cdot m}\big(U^{\!I\cdot m}\big)\right].
\end{equation}
\item The total variate of the subtree rooted at node $\!I$ defines an ancillary variate in the sup-tree rooted at the parent $\hat{\!I}$: 
\begin{equation}\label{eq:generalized-cv-reccurence}
\!I = \hat{\!I} \cdot k \quad \Rightarrow \quad U^{\hat{\!I} \cdot k} \coloneqq Z^{\!I}.
\end{equation}
\item In the case of a leaf node, $U^{\!I}$ that does not have any sub nodes, it is left alone.
\end{itemize}

Equations \eqref{eq:generalized-cv} and \eqref{eq:generalized-cv-reccurence} form a recursive relationship that defines the entire control variate structure for a tree, with the total variates at each root 
\begin{equation}
    Z^{(m)}, \quad 1 \leq m \leq M, 
\end{equation}
constituting a collection of high fidelity cases.

The total variate for the entire model forest, computed in the state space of model $m$
\begin{equation}
    \!Z^{(m)} = \sum_{m' = 1}^M w^{(m')} \tau^{(m, m')}(Z^{(m')}),
\end{equation}
corresponds to a weighted average of all the total variates in a similar fashion to~\cref{eq:model-averaging}.

\begin{theorem}[General Optimal Gain]\label{thm:general-optimal-gain}

    Consider the general control variate relation \cref{eq:generalized-cv}, with all the control variates $\widehat{U}^{\!I\cdot m}$, $m=1, \dots, M^{\!I}$ highly correlated to principal variate $X^{\!I}$ and to each other, and all ancillary variates $U^{\!I\cdot m}$, $m = 1, \dots, M^{\!I}$ independent of each other and on  the variate on the same fidelity and above.
    
    The set of optimal gains for \cref{eq:generalized-cv} is defined by the linear system,
       \begin{multline}\label{eq:general-optimal-gain-linear-equation}
            \begin{bmatrix}
                \*S^{\!I\cdot 1}\\ \vdots\\ \*S^{\!I\cdot M^{\!I}}
            \end{bmatrix}^T
            \begin{bmatrix}
                \Covs{g(\widehat{U}^{\!I\cdot 1})} + \Covs{g(U^{\!I\cdot 1})}  & \cdots    & \Cov{g(\widehat{U}^{\!I\cdot 1})}{g(\widehat{U}^{\!I\cdot M^{\!I})}}\\
                \vdots     & \ddots    & \vdots\\
                \Cov{g(\widehat{U}^{\!I\cdot M^{\!I}})}{g(\widehat{U}^{\!I\cdot 1})} & \cdots    & \Covs{g(\widehat{U}^{\!I\cdot M^{\!I}})} + \Covs{g(U^{\!I\cdot M^{\!I}})}
            \end{bmatrix}\\
            = 
            \begin{bmatrix}
                \Cov{h(X^{\!I \cdot 1})}{g(\widehat{U}^{\!I \cdot 1})}\\
                \vdots\\
                \Cov{h(X^{\!I \cdot M^{\!I}})}{g(\widehat{U}^{\!I \cdot M^{\!I}})}\\
            \end{bmatrix}^T.
    \end{multline}
\end{theorem}
\begin{proof}
    Without loss of generality, observe that for the control variate $m$ we can write $U^{\!I}$ as
    \begin{equation*}
        U^{\!I} = W^{\!I\cdot m} - \*S^{\!I\cdot m} \left(g(\widehat{U}^{\!I\cdot m}) - g(U^{\!I\cdot m})\right),
    \end{equation*}
    where the new term,
    \begin{equation*}
        W^{\!I\cdot m} = h(X^{\!I}) - \sum_{\substack{m'=1\\m'\not=m}}^{M^{\!I}} \*S^{\!I\cdot m'} \left(g(\widehat{U}^{\!I\cdot m'}) - g(U^{\!I\cdot m'})\right),
    \end{equation*}
    is considered to be the principal variate.
    Observe that by \cref{thm:optimal-gain}, the optimal gain can be written as,
    \begin{align*}
        \*S^{\!I\cdot m} &= \Cov{W^{\!I\cdot m}}{g(\widehat{U}^{\!I\cdot m})} {\left(\Cov{g(\widehat{U}^{\!I\cdot m})}{g(\widehat{U}^{\!I\cdot m})} + \Cov{g(U^{\!I\cdot m})}{g(U^{\!I\cdot m})}\right)}^{-1}\\
        &= \left(\Cov{X^{\!I}}{g(\widehat{U}^{\!I\cdot m})} - \sum_{\substack{m'=1\\m'\not=m}}^{M^{\!I}} \*S^{\!I\cdot m'}   \Cov{g(\widehat{U}^{\!I\cdot m'})}{g(\widehat{U}^{\!I\cdot m})} \right)\\
        &\phantom{=}\,\, \cdot {(\Cov{g(\widehat{U}^{\!I\cdot m})}{g(\widehat{U}^{\!I\cdot m})} + \Cov{g(U^{\!I\cdot m})}{g(U^{\!I\cdot m})})}^{-1}
    \end{align*}
    and taking the set of linear equations for each $m$, the solution is given by \cref{eq:general-optimal-gain-linear-equation}, as required.
\end{proof}

Thus, the linear control variate framework is extended to account for all the variables representing the model tree.

As in~\cref{sec:linear-cv-model-hier}, we make the assumptions that $h = \text{id}$ (the identity function), that all control variates are projections of the principal variate,
\begin{equation}
    \widehat{U}^{\!I \cdot m} = \theta^{\!I \cdot m}(X^{\!I}),
\end{equation}
and that the non-linear functions $g^{\!I\cdot m}$ are the interpolation operators,
\begin{equation}
    g^{\!I\cdot m} = \phi^{\!I\cdot m}.
\end{equation}

We now generalize \cref{thm:popov-gain}, for multiple control variates.
\begin{theorem}\label{thm:generalized-gain}
    Assume that:
    \begin{enumerate}[leftmargin = 32pt]
    \item The reconstruction errors are uncorrelated with the principal state:
    \begin{equation*}
        \Cov{X^{\!I}}{\left(X^{\!I} - \phi^{\!I\cdot m}(\widehat{U}^{\!I \cdot m})\right)} = \*0, \quad 1 \leq m \leq  M^{\!I}.
    \end{equation*}
    \item Covariances for full states reconstructed from the control and ancillary variates are identical,
    \begin{equation*}
        \Cov{\phi^{\!I\cdot m}(\widehat{U}^{\!I\cdot m})}{\phi^{\!I\cdot m}(\widehat{U}^{\!I\cdot m})} = \Cov{\phi^{\!I\cdot m}(U^{\!I\cdot m})}{\phi^{\!I\cdot m}(U^{\!I\cdot m})}, \quad 1 \leq m \leq  M^{\!I}.
    \end{equation*}
    \item All cross covariances,
    \begin{equation*}
        \Cov{\phi^{\!I\cdot m}(\widehat{U}^{\!I\cdot m})}{\phi^{\!I\cdot m'}(\widehat{U}^{\!I\cdot m'})}, \quad 1 \leq m, m' \leq  M^{\!I},
    \end{equation*}
    have the same dimension and are equal to each other.
    \end{enumerate}
   Under these assumptions the optimal gains are:
    \begin{equation}\label{eq:popov-gain-general}
        \*S^{\!I \cdot m} = \frac{1}{M^{\!I} + 1}\,\*I, \quad \forall m \in M^{\!I}.
    \end{equation}
\end{theorem}
\begin{proof}
    Under all the assumptions above, \cref{eq:general-optimal-gain-linear-equation} can be reduced to 
    \begin{equation*}
        \begin{bmatrix}
            \*S^{\!I\cdot 1}\\ \vdots\\ \*S^{\!I\cdot \Model{\!I}}
        \end{bmatrix}^T
        \begin{bmatrix}
            2\*I & \cdots & 1\*I\\
            \vdots & \ddots & \vdots\\
            1\*I & \cdots & 2\*I
        \end{bmatrix}\\
        = 
        \begin{bmatrix}
            \*I\\
            \vdots\\
            \*I\\
        \end{bmatrix}^T,
    \end{equation*}
    for which the solution, by inspection, is \cref{eq:popov-gain-general}.
\end{proof}

\Cref{thm:generalized-gain} can have some very important implications. First, as the number of models at a particular fidelity increases, each model's contribution decreases. For instance, when four models exist at a particular fidelity, each model's contribution would be scaled by $\frac{1}{5}$ by \cref{eq:popov-gain-general}. Second, as the number of models at a particular fidelity increases, the overall contribution of the lower fidelity information \textit{increases}. For instance the total contribution of the four models in the previous example is $\frac{4}{5}$. Thus, as the number of low fidelity models increases, our confidence in their total information increases as well.

\subsection{Forecasting with Model Forests}
\label{sec:model-propagation-for-mf}
We now turn our attention to state propagation through model forests. We combine the ideas introduced in \cref{sec:model-propagation-for-model-hier} for propagating model hierarchies  and \cref{sec:model-propagation-for-model-averages} for propagating model averages. As before, we need to propagate all variables in all trees and at all levels through their respective models.

As in~\cref{eq:linear-control-variate-model-propagation} and~\cref{eq:generalized-cv} we propagate each constituent variate and the total variate of model sub-tree,
\begin{gather}
X_{i+1}^{\!I} = \Model{\!I}(Z_i^{\!I}), \\
\widehat{U}_{i+1}^{\!I\cdot m} = \Model{\!I\cdot m}(\widehat{U}_i^{\!I\cdot m}), \\
U_{i+1}^{\!I\cdot m} = Z_{i+1}^{\!I\cdot m},\\
    Z_{i+1}^{\!I} \underset{\text{defined}}{=} X_{i+1}^{\!I} 
    - \sum_{m=1}^{M^{\!I}}\,\*S^{\!I\cdot m}\left[\phi^{\!I\cdot m} \left(\widehat{U}_{i+1}^{\!I\cdot m}\right) - \phi^{\!I\cdot m}\left(U_{i+1}^{\!I\cdot m}\right)\right],
\end{gather}
by making the assumption that the linear control variate framework applies the same after model propagation, with $Z^{(m)}$ for $1 < m < M$ serving as the base case.

With the total propagated variate of the whole forest in the space of model $m$ being,
\begin{equation}\label{eq:model-forest-total-variate}
    \Model{}(\!Z_i^{(m)}) = \!Z_{i+1}^{(m)} \coloneqq \sum_{m'=1}^M w^{(m')} \tau^{(m, m')}\left( Z^{(m')}_{i+1}\right),
\end{equation}
similar to the example in \cref{eq:model-average-propagation}.

The first two moments of~\cref{eq:model-forest-total-variate} can be written as,
\begin{gather}
    \Mean{\Model{}(\!Z^{(m)})} = \sum_{m'=1}^{M} w^{(m')}\Mean{\tau^{(m, m')}\left(\Model{(m')}(Z^{(m')})\right)},\\
    \Cov{\!Z^{(m)}}{\!Z^{(m)}} = \sum_{m'=1}^{M}\sum_{m''=1}^{M} w^{(m')} w^{(m'')} \Cov{\tau^{(m, m')}\left(\Model{(m')}(Z^{(m')})\right)}{\tau^{(m, m'')}\left(\Model{(m'')}(Z^{(m'')})\right)},
\end{gather}
which again faces the cross-covariance challenge from~\cref{sec:model-avaraging}.
In order to approximate the cross covariance term,
\begin{equation}
    \Cov{\Model{(m)}(Z^{(m)})}{\Model{(m')}(Z^{(m')})},
\end{equation}
by the method introduced in \cref{sec:model-propagation-for-model-averages}, we need access to ensembles of the total variates, $\En{Z^{(m)}}$ and $\En{Z^{(m')}}$.

We next provide one way to generate such samples.
\begin{theorem}\label{thm:optimal-gain-extension}
    For the bifidelity case,
    \begin{equation}
        Z^{(m)} = X^{(m)} - \sum_{m'=1}^M\*S^{(m,m')}\left(\phi^{(m,m')}(\widehat{U}^{(m,m')}) - \phi^{(m,m')}(U^{(m,m')})\right),
    \end{equation}
    under the assumptions of~\cref{thm:generalized-gain} and the strong assumption that
    \begin{equation}
        \Cov{X^{(m)}}{\phi^{(m,m')}(\widehat{U}^{(m,m')})} = \*I, \quad \forall m' = 1,\dots,M,
    \end{equation}
    the optimal gains associated with transforming the anomalies,
    \begin{equation}
        \An{Z^{(m)}} = \An{X^{(m)}} - \sum_{m'=1}^M \widetilde{\*S}^{(m,m')}\An{\phi^{(m, m')}\left(\widehat{U}^{(m,m')}\right)},
    \end{equation}
    from the principal and control variates to the total variate is,
    \begin{equation}\label{eq:optimal-gain-scaling}
        \widetilde{\*S}^{(m,m')} = \frac{1}{1+\sqrt{2 \slash (M+3)}}\,\*S^{(m,m')},
    \end{equation}
    which is a scalar multiple of the optimal gain in \cref{thm:generalized-gain}.
\end{theorem}
\begin{proof}
    Similar to~\cref{thm:general-optimal-gain}, we can write,
    \begin{equation*}
        Z^{(m)} = W^{(m,m')} - \*S^{(m,m')}\left(\phi^{(m,m')}(\widehat{U}^{(m,m')}) - \phi^{(m,m')}(U^{(m,m')})\right),
    \end{equation*}
    where the new principal variate is
    \begin{equation*}
        W^{(m,m')} = X^{(m)} - \sum_{\substack{m'' = 1\\m''\not=m'}}^M \*S^{(m, m'')} \left(\phi^{(m,m'')}(\widehat{U}^{(m,m'')}) - \phi^{(m,m'')}(U^{(m,m'')})\right).
    \end{equation*}
    Following equation (10) in \cite{whitaker2002ensemble}, the optimal gain modified to transform the anomalies can be written as
    \begin{multline*}
        \widetilde{\*S}^{(m,m')} = \*S^{(m,m')}\\ \cdot\left[ \*I + \left(\*I + \Cov{W^{(m,m')}}{\phi^{(m, m')}(\widehat{U}^{(m, m')})} 
        \Cov{\phi^{(m, m')}(U^{(m, m')})}{\phi^{(m, m')}(U^{(m, m')})}^{-1}\right)^{-1/2}\right]^{-1},
    \end{multline*}
    which, under the assumptions provided, simplifies to \cref{eq:optimal-gain-scaling} as required.
\end{proof}

As all ancillary variates are independent of the highest fidelity principal variate, \cref{thm:optimal-gain-extension} is readily extendable to all model forests, though we do not be explore such an extension in this paper.

We now show how an ensemble of the total variate can be generated.
\begin{corollary}\label{cor:Z-ensemble}
    Under linear assumptions on the projection and interpolation operators, and the assumption that
    \begin{equation}
        \En{\widehat{U}^{(m,m')}} = \theta^{(m,m')}\left(\En{X^{(m)}}\right),
    \end{equation}
    the control variates are transformations of the principal variate, an ensemble of samples from the total variate $Z^{(m)}$ can be written purely in terms of the principal variate $X^{(m)}$ as,
    \begin{multline}
        \En{Z^{(m)}} = \MeanE{Z}\,\*1_{N_X}^T \\+ \sum_{m'=1}^{M}\left[\En{X^{(m)}} - \widetilde{\*S}^{(m,m')}\phi^{(m,m')}\left(\En{\widehat{U}^{(m,m')}}\right)\right]\left(\*I_{N_{X^{(m)}}} - N_{X^{(m)}}^{-1} \*1_{N_{X^{(m)}}}\,\*1_{N_{X^{(m)}}}^T\right),
    \end{multline}
    where $\*1_{N_{X^{(m)}}}$ is a vector of $N_{X^{(m)}}$ ones, which assumes the gain in \cref{thm:generalized-gain}, and the total-principal variate relationship in \cref{eq:generalized-cv}.
\end{corollary}

We naturally assume the generalization of \cref{cor:Z-ensemble} to non-linear operators without analysis to its optimality. We can leverage \cref{cor:Z-ensemble} to compute the cross covariance terms of the model propagation through \cref{eq:covariance-ensemble-average}.

\section{Model Forest EnKF}
\label{sec:model-forest-enkf}

\begin{figure}
    \centering
    \includegraphics[width=0.95\linewidth]{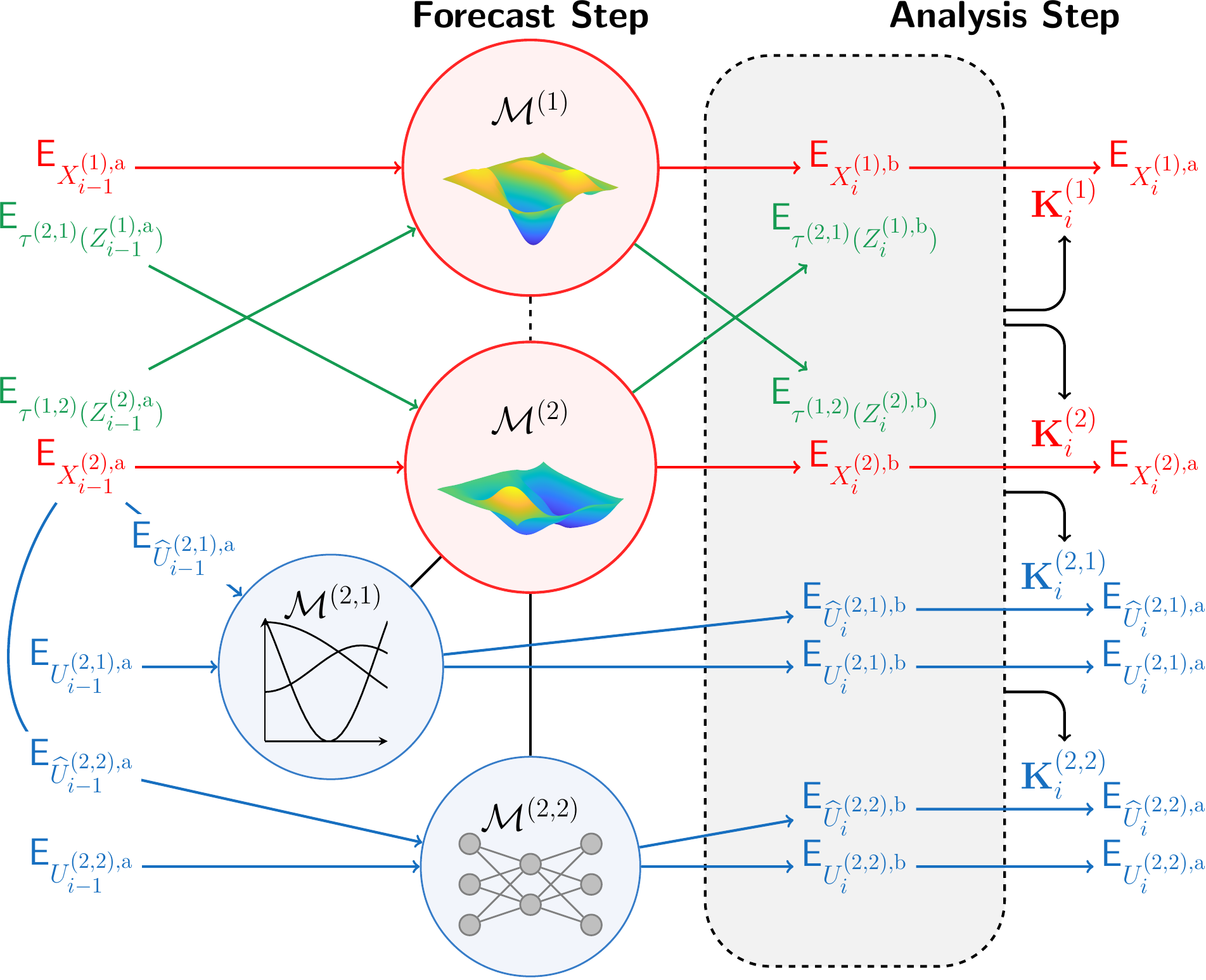}
    \caption[Model Forest EnKF]{One full step of the model forest ensemble Kalman filter. In this specific case, the model forest consists of two trees: one tree only has the high fidelity \mainmodel{}  model $\Model{(1)}$ and the other tree, \treetwosurr{}, has the high fidelity \mainmodel{} model $\Model{(2)}$ and two of its surrogates \surrmodel{}, $\Model{(2,1)}$ which is a reduced order model, and $\Model{(2,2)}$ which is a neural network-based model. In the forecast step, the high fidelity ensembles $\En{X^{(1),\|a}_{i-1}}$ and $\En{X^{(2),\|a}_{i-1}}$ are propagated through their respective models. The total variates of the trees represented in the space of the opposite tree, $\En{\tau^{(2,1)}(Z^{(1),\|a}_{i-1})}$ and $\En{\tau^{(1,2)}(Z^{(2),\|a}_{i-1})}$ are propagated through the models $\Model{(2)}$ and $\Model{(1)}$ respectively. The ensembles of control variates for the second tree, $\En{\widehat{U}^{(2,1),\|a}_{i-1}}$ and  $\En{\widehat{U}^{(2,2),\|a}_{i-1}}$ are generated from their principal variate ensemble and propagated through their respective surrogate models together with the ancillary variate ensembles $\En{U^{(2,1),\|a}_{i-1}}$ and $\En{U^{(2,2),\|a}_{i-1}}$. During the analysis step, the Kalman gains $\*K^{(1)}_i$, $\*K^{(2)}_i$, $\*K^{(2,1)}_i$, and $\*K^{(2, 2)}_i$ are generated and are used to propagate the principal, control and ancillary variates.}
    \label{fig:model-forest-enkf-example}
\end{figure}

We are now ready to introduce the model forest EnKF, combining elements from the multifidelity EnKF described in \cref{sec:multifilidety-enkf} and model forests described in \cref{sec:model-forests}. 
The model forest EnKF operates with two familiar steps: forecast and analysis. In the forecast step all the constituent ensembles, including ensembles of the total variate are propagated just like described in \cref{sec:model-propagation-for-mf}. 

For the analysis step, we largely  mirror the setup in~\cref{sec:multifilidety-enkf}.
Take the general control variate framework,
\begin{gather}
    \!Z^{(m),\|b}_i = \sum_{m'=1}^M w^{(m')}_i\tau^{(m, m')}_i Z^{(m'),\|b}_i, \label{eq:forest-total-variate}  \\
   \begin{multlined}[t] Z^{(m'),\|b}_i = X^{(m',\|b)}_i - \sum_{m''=1}^{M^{m'}} \*S^{(m')}_i \left[\phi^{(m',m'')}\left(\widehat{U}^{(m',m''),\|b} \right) - \phi^{(m',m'')}\left(U^{(m',m''),\|b} \right) \right],
   \end{multlined}
   \label{eq:tree-total-variate}
\end{gather}
that is induced by some model forest where each $\!Z^{(m),\|b}_i$ represents the total forest forecast in the space of model $m$, and each $Z^{(m'),\|b}_i$ represents the variable on the tree corresponding to the high fidelity model $m'$.

As in the multifidelity ensemble Kalman filter in~\cref{sec:multifilidety-enkf}, we assume the decomposition of the total variate~\cref{eq:forest-total-variate} recursively into the constituent ensembles, 
\begin{equation}\label{eq:model-forest-constituent-ensemble}
    \left(\En{\widehat{U}^{\!I,\|b}_i},\En{U^{\!I,\|b}_i}\right)_{\!I\in\!T},
\end{equation}
similar to~\cref{eq:constituent-ensembles},
where $\!T$ is the set of all tuples indexing the model forest.
The model forest ensemble Kalman filter operates on the constituent ensembles~\cref{eq:model-forest-constituent-ensemble} in a manner similar to~\cref{eq:multifidelity-Kalman-filter},
\begin{equation}
\begin{gathered}
    \En{U^{\!I, \|a}_i} = \En{U^{\!I, \|b}_i} - \*K^{\!I}_i\left[\!H^{\!I}_i\left(\En{U^{\!I, \|b}_i}\right) - \En{Y_i}\right],\\
    \En{\widehat{U}^{\!I, \|a}_i} = \En{\widehat{U}^{\!I, \|b}_i} - \*K^{\!I}_i\left[\!H^{\!I}_i\left(\En{\widehat{U}^{\!I, \|b}_i}\right) - \En{Y_i}\right],
\end{gathered}
\end{equation}
where, similar to~\cref{eq:low fidelity-observation-operator}, the observation operators can be recursively defined as:
\begin{equation}
    \!H^{\!I \cdot m}_i \coloneqq \!H^{\!I}_i \circ \phi^{\!I}_i, 
\end{equation}
and the Kalman gain $\*K^{\!I}$ contain all the information from the total variate~\cref{eq:forest-total-variate} in the space of model $\!I$ similar to~\cref{eq:fidelity-Kalman-gains}, requiring the use of~\cref{cor:Z-ensemble}.
    
After assimilation and before the next forecast step, the same heuristics (\cref{rem:MFEnKF-heuristics}) as in the multifildelity EnKF are applied: the means are corrected, and the control variate ensembles $\En{\widehat{U}^{\!I \cdot m, \|a}_i}$ are discarded and new control variate ensembles are generated from the principal variates,
\begin{equation}
    \En{\widehat{U}^{\!I \cdot m, \|a}_i} \xleftarrow[]{} \theta^{\!I \cdot m}\left(\En{U^{\!I, \|a}_i}\right),
\end{equation}
ensuring a strong correlation between the principal and control variates.
Inflation (\cref{rem:inflation}) is again required for the filter to converge. An illustration of the model forest EnKF is provided in~\cref{fig:model-forest-enkf-example}.

\section{Models}
\label{sec:models}

We now introduce the quasi-geostrophic (QG) equations ~\cite{foster2013finite,ferguson2008numerical,MW06,greatbatch2000four}, and two data-driven reduced-order surrogate models, one based on proper orthogonal decomposition (POD)~\cite{sirovich1987turbulence1, brunton2019data}, and one based on autoencoders (AE)~\cite{Popov2021b}.

\subsection{Quasi-Geostrophic equations}

QG equations are,
\begin{equation}
  \begin{split}
    \label{eq:QG}
    \omega_t + J(\psi,\omega) - {Ro}^{-1}\, \psi_x &= {Re}^{-1}\, \Delta\omega + {Ro}^{-1}\,F, \\
    J(\psi,\omega)&\equiv \psi_y\, \omega_x - \psi_x\, \omega_y,
  \end{split}
\end{equation}
where $\omega$ is vorticity, $\psi$ is the stream function, $Re$ is the Reynolds number, $Ro=0.0036$ is the Rossby number, $J$ is the Jacobian term, and $F$ is a symmetric double gyre forcing term,
$F = \sin\left(\pi(y-1)\right)$. The Reynolds number $Re$ is defined later. The vorticity term and the stream function are linear transformations $\omega = -\Delta\psi$, of each other. The spatial domain is $[0,1]\times[0,2]$ with homogeneous Dirichlet boundary conditions. A $63\times 127$ second order finite difference discretization is used. All calculations are performed on the streamfunction data.

More details about this model can be found in \cite{Popov2021a,Popov2021b}. The implementation used in this work is from the ODE Test Problems suite~\cite{otp,otpsoft}.

We now provide a brief description of the low-fidelity models that we construct. All our low-fidelity models are data-driven and intrusive, meaning that they rely both on a collection of data points, and on the original equations~\cref{eq:QG}. The data,
\begin{equation}\label{eq:data}
    \*X = \left[\*x_1, \*x_2, \dots, \*x_N\right]
\end{equation}
used to create the surrogates was generated by QG with $Re=450$, from a trajectory of $N=10^4$ data points spaced 30 days in model time ($\Delta t = 0.3268$) apart.

\subsubsection{Proper orthogonal decomposition ROM}
\label{sec:POD}
In POD, the projection and interpolation operators with be linear,
\begin{equation}\label{eq:POD-projection-and-interpolation}
    \theta(X) = \*\Phi^* X,\quad \phi(U) = \*\Phi U,
\end{equation}
where the matrix $\*\Phi \in \mathbb{R}^{n\times r}$ consists of the dominant $r$ eigenmodes of the second moment of the data~\cref{eq:data}, and $n$ is the dimension of the original data ($63\times 127$ for our QG implementation). For the POD model we take $r=25$ to be fixed representing a medium reduction in the dynamics.

The ROM itself is a quadratic dynamical system of the form,
\begin{equation}
    U_t = \*a + \*B U + U^T \!C U, 
\end{equation}
where $\*a$ is a vector, $\*B$ is a matrix and $\!C$ all defined in terms of the original equations~\cref{eq:QG} and the projection and interpolation operators~\cref{eq:POD-projection-and-interpolation}.
For more details on this reduced order model please see~\cite{Popov2021a}.

\subsubsection{Autoencoder-based ROM}
\label{sec:AE}

We now provide a brief overview of the AE based ROM.
Instead of linear operators~\cref{eq:POD-projection-and-interpolation}, we use two feed-forwards neural networks,
\begin{equation}
\begin{gathered}
    \theta(X) = \begin{multlined}[t]\*W^\theta_2\,\sigma(\*W^\theta_1 X + \*b^\theta_1) + \*b^\theta_2,\\ 
    \*W^\theta_1\in\mathbb{R}^{h\times n}, \*W^\theta_2\in\mathbb{R}^{r\times h}, \*b^\theta_1\in\mathbb{R}^{h}, \*b^\theta_2\in\mathbb{R}^{r},
    \end{multlined}\\
    \phi(U) = \begin{multlined}[t]\*C\left(\*W^\phi_2\,\sigma(\*W^\phi_1 U + \*b^\phi_1) + \*b^\phi_2\right),\\ \*W^\phi_1\in\mathbb{R}^{h \times r}, \*W^\phi_2\in\mathbb{R}^{n\times h}, \*b^\phi_1\in\mathbb{R}^{h}, \*b^\phi_2\in\mathbb{R}^{n},\label{eq:encoder-decoder-QG}
    \end{multlined} 
\end{gathered}
\end{equation}
to build non-linear projection and interpolation operators. Here $h=200$ is the fixed hidden dimension, $r=25$ is again the reduced dimension size just like for the POD ROM, and $\*C$ is a fixed convolutional layer performing spatial smoothing.

The cost function is a simple non-linear least squares cost over the data~\cref{eq:data},
\begin{equation}\label{eq:MLMFpoint-RMSE}
    \!L(\*X) = \sum_{i=1}^N \frac{1}{n}\norm{\*x_i - \phi(\theta(\*x_i))}_2^2 + \frac{\lambda}{r}\norm{\theta(\*x_i) - \theta(\phi(\theta(\*x_i)))}_2^2 ,
\end{equation}
with one extra term (with  parameter $\lambda=10^3$) added to ensure the right-invertability property,
\begin{equation}
    \theta = \theta\circ\phi\circ\theta,
\end{equation}
is weakly preserved.

The reduced order model itself is of the form
\begin{equation}
    U_t = \theta'(U)\,\*f\left(\phi(U)\right),
\end{equation}
where $\*f$ represents the streamfunction dynamics of QG~\cref{eq:QG}, and $\theta'(U)$ is the derivative of the encoder~\cref{eq:encoder-decoder-QG}.
For more details on this reduced order model please see~\cite{Popov2021b}.

\section{Numerical Experiments}
\label{sec:numerical-experiments}

The goal of our numerical experiments is to show proof-of-concept: that both model trees and model forests impart some sort of advantage, whether in accuracy or computational cost.

For our true natural system we take QG with Reynolds number of $Re=450$.
We perform sequential data assimilation experiments over a forecast period of one day ($\Delta t = 0.0109$ in model time), measuring $150$ evenly-spaced points with observation error covariance of $\Cov{Y}{Y} = \*I_{150}$. For an accuracy metric, we analyze the total variate mean analysis spatio-temporal root mean square error,
\begin{equation}\label{eq:RMSE}
    \text{RMSE}(\MeanE{Z^\|a}, X^\|t) = \sqrt{\frac{1}{T n}\sum_{i=T_0}^{T_f} \sum_{k=1}^n \left(\left[\MeanE{Z^\|a_i}\right]_k - \left[X^t_i\right]_k\right)^2},
\end{equation}
over the time indices $T_0 = 51$ to $T_f = 350$, with $T = T_f-T_0$, discarding the first 50 steps for spinup. The RMSE~\cref{eq:RMSE} is averaged over 20 independent model runs, with different realization of the initial conditions of both the natural model and ensembles.

\subsection{Model tree experiment}
\label{sec:model-tree-experiment}

\begin{figure}[t]
    \begin{center}
     \includegraphics[width=1\linewidth]{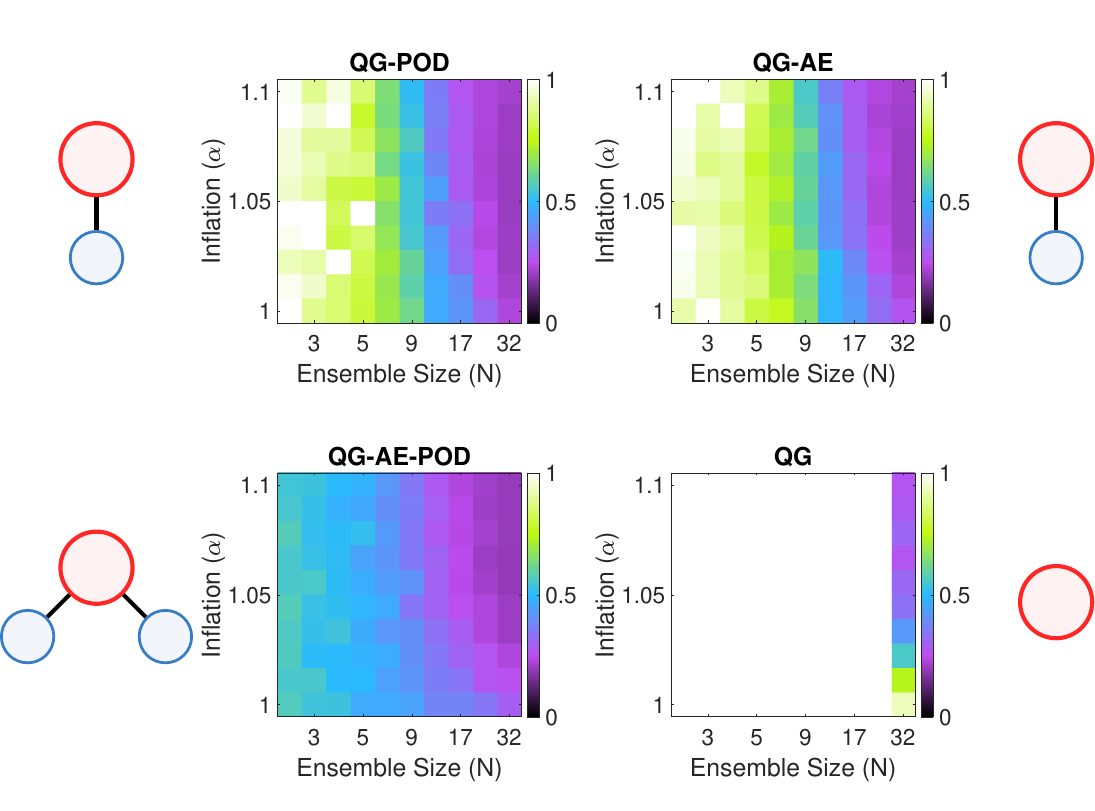}
    \end{center}
\caption[Model Tree Experiment]{Analysis RMSE for choices of ensemble size $N$ for the principal variable $X^{(1)}$ and inflation $\alpha$ for the principal variate $X^{(1)}$ for various model trees. The top left figure represents a bifidelity tree, \treeonesurr{} with QG as the high fidelity, \mainmodel{}, and the POD ROM as the low fidelity, \surrmodel{}, model. The top right figure represents a bifidelity tree \treeonesurr{} with QG as the high fidelity, \mainmodel{}, and the AE ROM as the low fidelity, \surrmodel{}, model. The bottom left figure represents a bifidelity tree, \treetwosurr{}, with QG as the high fidelity, \mainmodel{}, the POD ROM as one of the low fidelity, \surrmodel{}, models, and the AE ROM as the other low fidelity, \surrmodel{}, model. The bottom right figure represents a unifidelity tree, \mainmodel{}, consisting of just QG.}
\label{fig:model-tree-experiment}
\end{figure}

Our first experiment aims to show that a simple model forest with two low-fidelity models, \treetwosurr{}, is advantageous to use over bifidelity model hierarchies, \treeonesurr{}. 

We experiment on the following model trees: 
\begin{enumerate}
    \item the unifidelity tree, \mainmodel{}, with QG with $Re=450$ as the high fidelity model $\Model{(1)}$ that is the the same model as nature,
    \item the bifidelity tree, \treeonesurr{}, with QG with $Re=450$ as the high fidelity model $\Model{(1)}$ and the POD model as its surrogate model $\Model{(1,1)}$, and
    \item the bifidelity tree, \treeonesurr{}, with QG with $Re=450$ as the high fidelity model $\Model{(1)}$ and the AE model as its surrogate model $\Model{(1,1)}$, and
    \item the bifidelity tree, \treetwosurr{}, with QG with $Re=450$ being the high fidelity model $\Model{(1)}$ with the POD model as its surrogate model $\Model{(1,1)}$, and the AE model as its surrogate model $\Model{(1,2)}$,
\end{enumerate}
representing a minimal proof-of-concept of model trees.

We fix the reduced ensemble sizes to a low $N_{U^{(1,1)}} = N_{U^{(1,2)}} = 12$, and the reduced inflation to $\alpha_{U^{(1,1)}} = \alpha_{U^{(1,2)}} = 1.05$, and vary the high fidelity ensemble size $N$ logarithmically from the list $2, 3, 4, 5, 7,9, 13, 17, 24, 32$ and high fidelity inflation $\alpha$ linearly  in the range $[1, 1.1]$. Calculating the total variate mean analysis spatio-temporal root mean square error through~\cref{eq:RMSE} the results of the experiment can be seen in figure~\cref{fig:model-tree-experiment}.

As can be seen from the results, a high fidelity ensemble size of $N=32$ is required for the unifidelity, \mainmodel{}, filter to converge with QG. Both the one surrogate bifidelity trees, \treeonesurr{}, with either the POD or AE models significantly reduced the high fidelity ensemble size requirements, with as little as $N=7$ high fidelity ensemble members required for convergence, and showing results as accurate as the unifidelity, \mainmodel{}, filter for a high fidelity ensemble size of $N=17$, cutting the high fidelity model runs required in half. 
A very surprising result is that the two surrogate bifidelity tree, \treetwosurr{}, shows good stability behavior even for a high fidelity ensemble size of $N=2$ with good accuracy for an ensemble size of $N=13$, meaning that significantly less high fidelity ensemble members are required to ensure confidence in the filter results.

\subsection{Model forest experiment}
\label{sec:model-forest-experiment}

\begin{figure}[t]
    \begin{center}
        \includegraphics[width=0.825\linewidth]{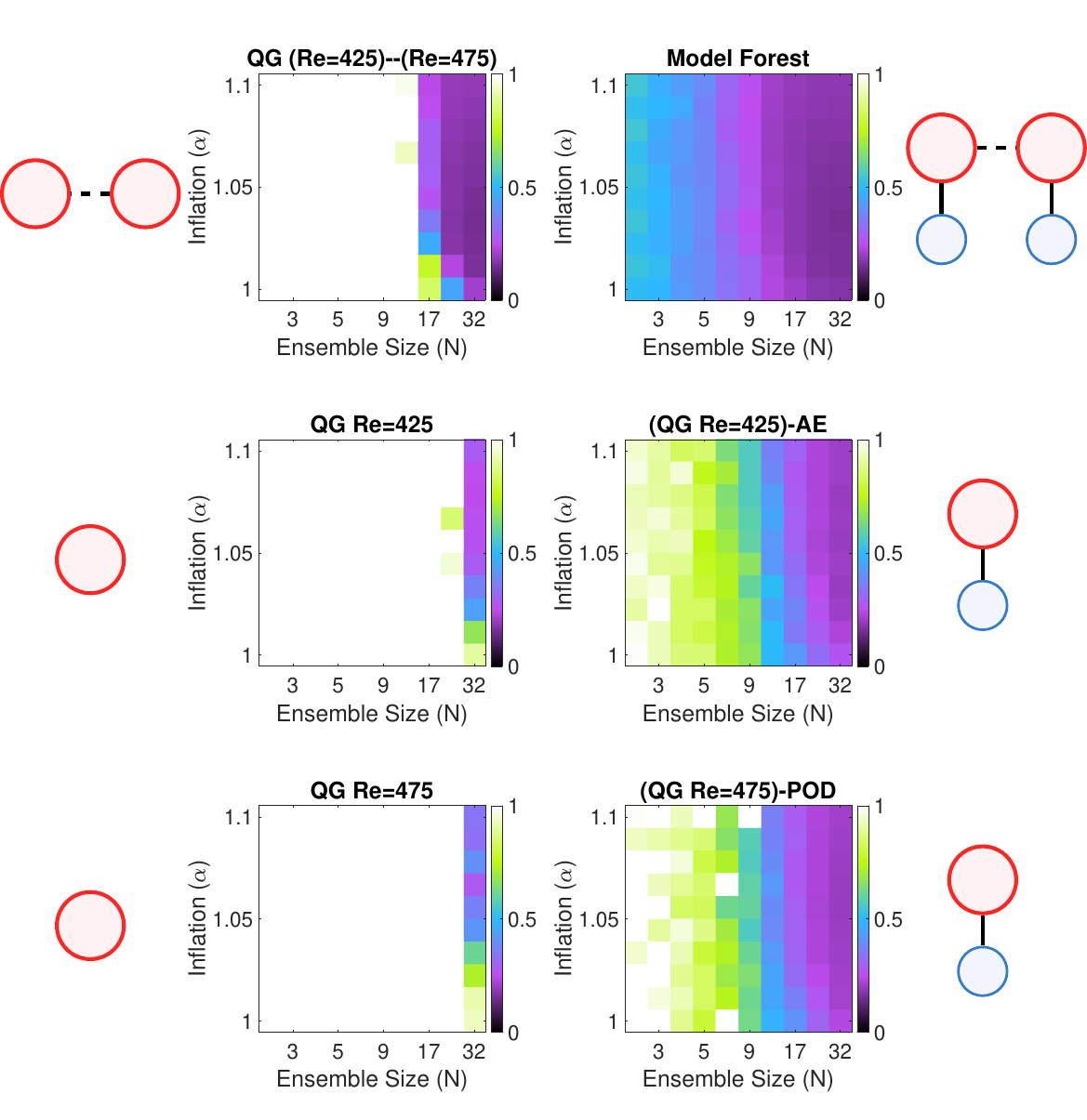}
    \end{center}
    \caption[Model Forest Experiment]{Analysis RMSE for choices of ensemble size $N$ and inflation $\alpha$ for the principal variate $X^{(1)}$ for variate model forests. The top left figure represents a model average, \modelaverage{}, of two QG models with Reynolds numbers of $Re=425$ and $Re=475$.  The second row left figure represents a unifidelity tree, \mainmodel{}, consisting of just QG with Reynolds number $Re=425$. The bottom left figure represents a unifidelity tree, \mainmodel{}, consisting of just QG with Reynolds number $Re=475$.
    The top right figure represents a model forest, \forestonesurr{} of two bifidelity trees, \treeonesurr{}, with the first tree consisting of QG  with Reynolds number $Re=425$ as the high fidelity, \mainmodel{}, and the AE ROM as the low fidelity, \surrmodel{}, model, and the second tree consisting of QG with Reynolds number $Re=475$ as the high fidelity, \mainmodel{}, and the POD ROM as the low fidelity, \surrmodel{}, model. The second row right figure represents a bifidelity tree, \treeonesurr{}, with QG with Reynolds number $Re=425$ as the high fidelity, \mainmodel{}, and the AE ROM as the low fidelity, \surrmodel{}, model. The bottom right figure represents a bifidelity tree, \treeonesurr{}, with QG with Reynolds number $Re=475$ as the high fidelity, \mainmodel{}, and the POD ROM as the low fidelity, \surrmodel{}, model.}
    \label{fig:model-forest-experiment}
\end{figure}

Our second experiment illustrates the usefulness of both model averaging and model forests. 
Recall that we take QG~\cref{eq:QG} with Reynolds number $Re=450$ as the natural ground truth.
Instead of assuming that our high fidelity model is the same as nature, the twin experiment assumption, we assume that that are now two competing models of nature, between which we cannot decide. The first of these models with be QG with Reynolds number $Re=425$ and the second with $Re=475$. This ensures that the two competing models are close enough to the behavior of the original model, but are not exact.

The POD and AE surrogate models are still trained on the data generated with $Re=450$, as we assume that data-driven models have access to some natural data, thus more accurately capture the natural behavior of the ground truth.

We experiment with the following model forests:
\begin{enumerate}
    \item the unifidelity tree, \mainmodel{}, with QG with $Re=425$ as the high fidelity model $\Model{(1)}$,
    \item the unifidelity tree, \mainmodel{}, with QG  with $Re=475$ as the high fidelity model $\Model{(1)}$,
    \item the model average, \modelaverage{}, of QG with $Re=425$ as the high fidelity model $\Model{(1)}$ and QG with $Re=475$ as the high fidelity model $\Model{(2)}$,
    \item the bifidelity tree, \treeonesurr{}, with QG with $Re=425$ as the high fidelity model  $\Model{(1)}$ and the AE model trained on $Re=450$ as its surrogate $\Model{(1,1)}$ ,
    \item the bifidelity tree, \treeonesurr{}, with QG with $Re=475$ as the high fidelity model $\Model{(1)}$ and the POD model trained on $Re=450$ as its surrogate $\Model{(1,1)}$, and
    \item the model forest, \forestonesurr{}, as the average, with equal weights $w^{(1)} = w^{(2)} = \frac{1}{2}$, of the above two bifidelity trees, thus QG with $Re=425$ is the high fidelity model $\Model{(1)}$, QG with $Re=275$ is the high fidelity model $\Model{(2)}$, the AE model is $\Model{(1,1)}$, and the POD model is $\Model{(2,1)}$.
\end{enumerate}
representing the minimal proof-of-concept for model forests.

We again fix the reduced ensemble sizes to a low $N_{U^{(1,1)}} = N_{U^{(2,1)}} = 12$, and the reduced inflation to $\alpha_{U^{(1,1)}} = \alpha_{U^{(2,1)}} = 1.05$, and vary the high fidelity ensemble size $N$ logarithmically from the list $2, 3, 4, 5, 7, 9, 13, 17, 24, 32$ and high fidelity inflation $\alpha$ linearly  in the range $[1, 1.1]$. Calculating the total variate mean analysis spatio-temporal root mean square error through~\cref{eq:RMSE} the results of the experiment can be seen in figure~\cref{fig:model-forest-experiment}.

Recall that for a model average or a model forest, the number of high fidelity model runs rises exponentially because of the need to compute cross covariances~\cref{eq:covariance-ensemble-average}. Because of this, a model average, \modelaverage{}, with two models both of which have $N$ ensemble members, would require $4N$ high fidelity model runs to compute the covariance information. The results are nevertheless encouraging. The unifidelity, \mainmodel{}, filter with QG with Reynolds number $Re=475$ performed significantly worse than in the previous experiment in \cref{sec:model-tree-experiment}, however coupling it with an accurate POD model in a bifidelity structure, \treeonesurr{}, significantly improved its performance. The model average, \modelaverage{}, requires a minimum of $4\times 17$ high fidelity model runs to be stable, thus not being a feasible alternative to any one individual model, but the model forest, \forestonesurr{}, provides fairly accurate results for $4\times 7$ high fidelity model runs, while also being stable for much smaller ensemble sizes. If the number of high fidelity model runs used to compute cross-covariance can be reduced in an efficient manner, it is the authors' belief that the model forest approach is the future of ensemble filtering algorithms.

Throughout this whole work, we have ignored the discussion of model error. Model error assumptions are necessary in the unifidelity, \mainmodel{}, case when the model is not an exact representation of nature. But the results from the model forest experiment were obtained without any model error assumptions, and yet yield an impressive level of accuracy. We hypothesize that instead of attempting to compensate for model error through process noise, like is common with EnKF based methods, model forests can serve as an alternative.

\section{Conclusions}
\label{sec:conclusions}

This work introduces model forests, a concept that generalizes model hierarchies,and formalizes situations where collections of models can be used in a rigorous systematic way in data assimilation. We show how random variables attached to these model forests can be used to propagate information from the models contained within through a generalized theory of control variates.
Potential advantages of working with full model forests instead of one model include higher accuracy in the case of multiple competing inaccurate models, and a non-trivial reduction in computational cost with no loss of accuracy.

Using this idea, we extended the multifidelity ensemble Kalman filter to the model forest ensemble Kalman filter, replacing the MFEnKF acronym. Through various numerical experiments on many different combinations of model tree and model forest, we have shown that the MFEnKF not only significantly reduced the need for high fidelity model runs, but also has the potential to replace assumptions about model error, as an ensemble of models could potentially approximate our uncertainty about the model propagation.

There are many venues to pursue in future research. Extending the MFEnKF family of algorithms to square-root filters (see \cite{asch2016data} for an in-depth look at all the different variations) is a requirement for the filter to be used in operational settings. Extending model forests to the ensemble transport particle filter~\cite{reich2013nonparametric} would potentially allow for the use of particle filters with higher-dimensional models.
An alternative approach is to construct a model forest ensemble variational Fokker-Plank filter~\cite{subrahmanya2021ensemble} which allows for the use of general classes of parameterized distributions freeing the filter from Gaussian assumptions.

A different future research venue is constructing data-driven reduced order algorithms tailored specifically for use in a model forest. This would potentially open the door for more efficient and accurate surrogate models and even more significantly reduce the number of high fidelity model runs.

\bibliographystyle{siamplain}
\bibliography{Bib/biblio,Bib/traian,Bib/sandu,Bib/data_assim_multilevel,Bib/data_assim_kalman}

\end{document}

%% file: shared.tex
\usepackage{cite}[sort]

\usepackage{tkz-graph}

\usetikzlibrary{calc}

\usepackage{pgfplots}

\pgfplotsset{compat=1.16}

\input{header.tex}

\input{usefulstuff.tex}

\graphicspath{ {Figures/} }

\def\realtitle{The Model Forest Ensemble Kalman Filter}

\title{{\realtitle}\thanks{Submitted to the ArXiv \today.
		\funding{The work of Popov and Sandu was supported by DOE through award ASCR DE-SC0021313, by NSF  through award  CDS\&E--MSS 1953113, and by the Computational Science Laboratory at Virginia Tech.}}}
\date{\today}

\author{Andrey A. Popov\textsuperscript{, }\footnotemark[2]\ \footnotemark[3]
\and Adrian Sandu\footnotemark[3]
}

\tikzset{MainModel/.style={circle, draw=red!85, fill=red!5, very thick, minimum size=22mm},
    SurrModel/.style={circle, draw=RoyalBlue!85, fill=RoyalBlue!5, thick, minimum size=16mm},
    MainModelSmall/.style={circle, draw=red!85, fill=red!5, thick, minimum size=0.5em, inner sep=0pt}, 
    SurrModelSmall/.style={circle, draw=RoyalBlue!85, fill=RoyalBlue!5, semithick, minimum size=0.3em, inner sep=0pt},
    MainModelMedium/.style={circle, draw=red!85, fill=red!5, very thick, minimum size=7.33mm},
    SurrModelMedium/.style={circle, draw=RoyalBlue!85, fill=RoyalBlue!5, thick, minimum size=5.33mm},}

\tikzset{MainModel/.style={circle, draw=red!85, fill=red!5, very thick, minimum size=22mm},
    SurrModel/.style={circle, draw=RoyalBlue!85, fill=RoyalBlue!5, thick, minimum size=16mm},
    MainModelSmall/.style={circle, draw=red!85, fill=red!5, thick, minimum size=0.5em, inner sep=0pt}, 
    SurrModelSmall/.style={circle, draw=RoyalBlue!85, fill=RoyalBlue!5, semithick, minimum size=0.3em, inner sep=0pt}}

\newcommand{\mainmodel}{\raisebox{0em}{\begin{tikzpicture}
            \node[MainModelSmall] (t1root) at (0,0) {};
\end{tikzpicture}}}

\newcommand{\surrmodel}{\raisebox{0em}{\begin{tikzpicture}
            \node[SurrModelSmall] (t1root) at (0,0) {};
\end{tikzpicture}}}

\newcommand{\modelaverage}{\raisebox{0em}{\begin{tikzpicture}
            \node[MainModelSmall] (t1root) at (0,0) {};
            \node[MainModelSmall] (t2root) at (1.35em,0) {};
            \draw[very thick, dotted] (t1root) -- (t2root);
\end{tikzpicture}}}

\newcommand{\treeonesurr}{\raisebox{-0.28em}{\begin{tikzpicture}
            \node[MainModelSmall] (t1root) at (0,0) {};
            \node[SurrModelSmall] (t1sur1) at (0,-0.6em) {};
            \draw[very thick] (t1root) -- (t1sur1);
\end{tikzpicture}}}

\newcommand{\treetwosurr}{\raisebox{-0.28em}{\begin{tikzpicture}
            \node[MainModelSmall] (t1root) at (0,0) {};
            \node[SurrModelSmall] (t1sur1) at (-0.5em,-0.5em) {};
            \node[SurrModelSmall] (t1sur2) at (0.5em,-0.5em) {};
            \draw[very thick] (t1root) -- (t1sur1);
            \draw[very thick] (t1root) -- (t1sur2);
\end{tikzpicture}}}

\newcommand{\forestonesurr}{\raisebox{-0.28em}{\begin{tikzpicture}
            \node[MainModelSmall] (t1root) at (0,0) {};
            \node[SurrModelSmall] (t1sur1) at (0,-0.6em) {};
            \draw[very thick] (t1root) -- (t1sur1);
            
            \node[MainModelSmall] (t2root) at (1.35em,0) {};
            \node[SurrModelSmall] (t2sur1) at (1.35em,-0.6em) {};
            \draw[very thick] (t2root) -- (t2sur1);
            
            \draw[very thick, dotted] (t1root) -- (t2root);
\end{tikzpicture}}}

%% file: header.tex
\usepackage{bbm,mathrsfs,mathtools,amsmath,amssymb,amsfonts,stmaryrd,upgreek}
\usepackage{fancyhdr,subfigure}
\usepackage{epsfig,graphicx,wrapfig,fancyvrb,listings,color,textcomp,float}
\usepackage{algorithm}
\Crefname{ALC@unique}{Line}{Lines}
\usepackage{enumitem,url,hyperref}

\usepackage{tikz}
\usepackage{tikz-qtree}
\usetikzlibrary{positioning}
\usetikzlibrary{calc}

\usepackage[listings,skins,breakable]{tcolorbox}
{%
\end{tcolorbox}\endgroup}

\newif\ifcomplete

\usepackage{accents}
\newcommand*{\bigdot}[1]{\accentset{\mbox{\large\bfseries .}}{#1}}

\newcommand{\x}[1][]{%
   \ifthenelse{ \equal{#1}{} }
      {X}
      {X^{[#1]}}}

\newcommand{\uh}[1][]{%
   \ifthenelse{ \equal{#1}{} }
      {\hat{U}}
      {\hat{u}^{(#1)}}}
\newcommand{\xb}[1][]{%
   \ifthenelse{ \equal{#1}{} }
      {X^{\mathrm{b}}}
      {\mathbf{X}^{\mathrm{b},[#1]}}}
      
\newcommand{\ub}[1][]{%
   \ifthenelse{ \equal{#1}{} }
      {U^{\mathrm{b}}}
      {\mathbf{U}^{\mathrm{b},[#1]}}}
      
\newcommand{\ua}[1][]{%
   \ifthenelse{ \equal{#1}{} }
      {U^{\mathrm{a}}}
      {\mathbf{U}^{\mathrm{a},[#1]}}}
      
\newcommand{\uha}[1][]{%
   \ifthenelse{ \equal{#1}{} }
      {\hat{U}^{\mathrm{a}}}
      {\mathbf{\hat{U}}^{\mathrm{a},[#1]}}}
      
\newcommand{\pub}[1][]{%
   \ifthenelse{ \equal{#1}{} }
      {\*\Phi_r\,U^{\mathrm{b}}}
      {\*\Phi_r\,\mathbf{U}^{\mathrm{b},[#1]}}}
      
\newcommand{\uhb}[1][]{%
   \ifthenelse{ \equal{#1}{} }
      {\hat{U}^{\mathrm{b}}}
      {\mathbf{\hat{U}}^{\mathrm{b},[#1]}}}
      
\newcommand{\puhb}[1][]{%
   \ifthenelse{ \equal{#1}{} }
      {\*\Phi_r\,\hat{U}^{\mathrm{b}}}
      {\*\Phi_r\,\mathbf{\hat{U}}^{\mathrm{b},#1}}}
      
\newcommand{\zb}[1][]{%
   \ifthenelse{ \equal{#1}{} }
      {Z^{\mathrm{b}}}
      {\mathbf{Z}^{\mathrm{b},#1}}}
      
\newcommand{\za}[1][]{%
   \ifthenelse{ \equal{#1}{} }
      {Z^{\mathrm{a}}}
      {\mathbf{Z}^{\mathrm{a},[#1]}}}

\newcommand{\xa}[1][]{%
   \ifthenelse{ \equal{#1}{} }
      {X^{\rm a}}
      {\mathbf{X}^{\mathrm{a},[#1]}}}
\newcommand{\xf}[1][]{%
   \ifthenelse{ \equal{#1}{} }
      {\mathbf{x}^{\rm f}}
      {\mathbf{x}^{{\rm f}[#1]}}}

\newcommand{\hofx}{\mathbf{z}}

\newcommand{\hofxb}[1][]{%
   \ifthenelse{ \equal{#1}{} }
      {\hofx^{\mathrm{b}}}
      {\hofx^{{\mathrm{b}}[#1]}}}
\newcommand{\hofxa}[1][]{%
   \ifthenelse{ \equal{#1}{} }
      {\hofx^{\rm a}}
      {\hofx^{{\rm a}[#1]}}}
\newcommand{\dhofxb}[1][]{%
   \ifthenelse{ \equal{#1}{} }
      {\bigdot{\hofx}^{\mathrm{b}}}
      {\bigdot{\hofx}^{{\mathrm{b}}[#1]}}}
\newcommand{\y}{ \mathbf{y} }

\newcommand{\z}{ \mathbf{z} }

\renewcommand{\v}{ \mathbf{v} }

\newcommand{\errb}[1][]{%
   \ifthenelse{ \equal{#1}{} }
      {\varepsilon^{\mathrm{b}}}
      {\varepsilon^{{\mathrm{b}}[#1]}}}
\newcommand{\erra}[1][]{%
   \ifthenelse{ \equal{#1}{} }
      {\varepsilon^{\rm a}}
      {\varepsilon^{{\rm a}[#1]}}}
\newcommand{\erro}[1][]{%
   \ifthenelse{ \equal{#1}{} }
      {\eta}
      {\boldsymbol{\eta}^{(#1)}}}

\newcommand{\errm}[1][]{%
   \ifthenelse{ \equal{#1}{} }
      {\eta}
      {\eta^{[#1]}}}
\newcommand{\wb}[1][]{%
   \ifthenelse{ \equal{#1}{} }
      {w^{\mathrm{b}}}
      {w^{{\mathrm{b}}[#1]}}}
\newcommand{\wa}[1][]{%
   \ifthenelse{ \equal{#1}{} }
      {w^{\rm a}}
      {w^{{\rm a}[#1]}}}

\newtheorem{remark}{Remark}

\newcommand{\norm}[1]{\bigl\Vert #1 \bigr\Vert}



%% file: usefulstuff.tex
\def\!#1{\mathcal{#1}}
\def\*#1{\boldsymbol{\mathbf{#1}}}
\def\|#1{\textnormal{#1}}

\def\norm#1{\left\lVert#1\right\rVert}

\newcommand{\Mean}[1]{\*\mu_{#1}}
\newcommand{\MeanE}[1]{\widetilde{\*\mu}_{#1}}
\newcommand{\Cov}[2]{\*\Sigma_{#1,#2}}
\newcommand{\Covs}[1]{\*\Sigma_{#1}}
\newcommand{\CovE}[2]{\widetilde{\*\Sigma}_{#1,#2}}

\newcommand{\En}[1]{\mathsf{E}_{#1}}
\newcommand{\An}[1]{\mathsf{A}_{#1}}

\DeclareMathOperator{\tr}{tr}

\newcommand{\Model}[1]{\mathcal{M}^{#1}}

%% file: main_forest.bbl
\begin{thebibliography}{10}

\bibitem{asch2016data}
{\sc M.~Asch, M.~Bocquet, and M.~Nodet}, {\em {Data assimilation: methods,
  algorithms, and applications}}, SIAM, 2016.

\bibitem{brunton2019data}
{\sc S.~L. Brunton and J.~N. Kutz}, {\em Data-driven science and engineering:
  Machine learning, dynamical systems, and control}, Cambridge University
  Press, 2019.

\bibitem{Burgers_1998_EnKF}
{\sc G.~Burgers, P.~J. van Leeuwen, and G.~Evensen}, {\em Analysis scheme in
  the {E}nsemble {K}alman {F}ilter}, Monthly Weather Review, 126 (1998),
  pp.~1719--1724.

\bibitem{chada2020multilevel}
{\sc N.~K. Chada, A.~Jasra, and F.~Yu}, {\em Multilevel ensemble {Kalman-Bucy}
  filters}, arXiv preprint arXiv:2011.04342,  (2020).

\bibitem{chernov2020multilevel}
{\sc A.~Chernov, H.~Hoel, K.~J. Law, F.~Nobile, and R.~Tempone}, {\em
  Multilevel ensemble kalman filtering for spatio-temporal processes},
  Numerische Mathematik,  (2020), pp.~1--55.

\bibitem{otpsoft}
{\sc {Computational Science Laboratory}}, {\em {ODE} test problems}, 2020,
  \url{https://github.com/ComputationalScienceLaboratory/ODE-Test-Problems}
  (accessed 2020-01-16).

\bibitem{donoghueamulti}
{\sc G.~Donoghuea and M.~Yanoa}, {\em A multi-fidelity ensemble {Kalman} filter
  with hyperreduced reduced-order models}, 2022,
  \url{http://arrow.utias.utoronto.ca/~myano/papers/donoghue_yano_2022_multifidelity_enkf.pdf}.

\bibitem{dormann2018model}
{\sc C.~F. Dormann, J.~M. Calabrese, G.~Guillera-Arroita, E.~Matechou, V.~Bahn,
  K.~Barto{\'n}, C.~M. Beale, S.~Ciuti, J.~Elith, K.~Gerstner, et~al.}, {\em
  Model averaging in ecology: A review of bayesian, information-theoretic, and
  tactical approaches for predictive inference}, Ecological Monographs, 88
  (2018), pp.~485--504.

\bibitem{evensen2022data}
{\sc G.~Evensen, F.~C. Vossepoel, and P.~J. van Leeuwen}, {\em Data
  Assimilation Fundamentals: A Unified Formulation of the State and Parameter
  Estimation Problem}, Springer Nature, 2022.

\bibitem{eyring2016overview}
{\sc V.~Eyring, S.~Bony, G.~A. Meehl, C.~A. Senior, B.~Stevens, R.~J. Stouffer,
  and K.~E. Taylor}, {\em Overview of the coupled model intercomparison project
  phase 6 (cmip6) experimental design and organization}, Geoscientific Model
  Development, 9 (2016), pp.~1937--1958.

\bibitem{ferguson2008numerical}
{\sc J.~Ferguson}, {\em {A numerical solution for the barotropic vorticity
  equation forced by an equatorially trapped wave}}, master's thesis,
  University of Victoria, 2008.

\bibitem{foster2013finite}
{\sc E.~L. Foster, T.~Iliescu, and Z.~Wang}, {\em A finite element
  discretization of the streamfunction formulation of the stationary
  quasi-geostrophic equations of the ocean}, Comput. Methods Appl. Mech.
  Engrg., 261 (2013), pp.~105--117.

\bibitem{giles2008multilevel}
{\sc M.~B. Giles}, {\em Multilevel monte carlo path simulation}, Operations
  Research, 56 (2008), pp.~607--617.

\bibitem{giles2015multilevel}
{\sc M.~B. Giles}, {\em Multilevel monte carlo methods}, Acta Numerica, 24
  (2015), pp.~259--328.

\bibitem{mfnets}
{\sc A.~Gorodetsky, J.~D. Jakeman, and G.~Geraci}, {\em {MFNets}: Data
  efficient all-at-once learning of multifidelity surrogates as directed
  networks of information sources}, 2020,
  \url{https://doi.org/10.48550/ARXIV.2008.02672},
  \url{https://arxiv.org/abs/2008.02672}.

\bibitem{greatbatch2000four}
{\sc R.~J. Greatbatch and B.~T. Nadiga}, {\em Four-gyre circulation in a
  barotropic model with double-gyre wind forcing}, J. Phys. Oceanogr., 30
  (2000), pp.~1461--1471.

\bibitem{Hoel_2016_MLEnKF}
{\sc H.~Hoel, K.~J.~H. Law, and R.~Tempone}, {\em Multilevel ensemble {Kalman}
  filtering}, SIAM Journal on Numerical Analysis, 54 (2016),
  \url{https://doi.org/10.1137/15M100955X}.

\bibitem{hoel2019multilevel}
{\sc H.~Hoel, G.~Shaimerdenova, and R.~Tempone}, {\em Multilevel ensemble
  {Kalman} filtering based on a sample average of independent enkf estimators},
  Foundations of Data Science,  (2020), pp.~351--390.

\bibitem{hoel2021multi}
{\sc H.~Hoel, G.~Shaimerdenova, and R.~Tempone}, {\em Multi-index ensemble
  {Kalman} filtering}, arXiv preprint arXiv:2104.07263,  (2021).

\bibitem{law2015data}
{\sc K.~Law, A.~Stuart, and K.~Zygalakis}, {\em {Data assimilation: a
  mathematical introduction}}, vol.~62, Springer, 2015.

\bibitem{MW06}
{\sc A.~J. Majda and X.~Wang}, {\em Nonlinear dynamics and statistical theories
  for basic geophysical flows}, Cambridge University Press, Cambridge, 2006.

\bibitem{petersen2008matrix}
{\sc K.~Petersen, M.~Pedersen, et~al.}, {\em The matrix cookbook}, Technical
  University of Denmark, 15 (2008).

\bibitem{Popov2021a}
{\sc A.~A. Popov, C.~Mou, A.~Sandu, and T.~Iliescu}, {\em A multifidelity
  ensemble {Kalman} filter with reduced order control variates}, SIAM Journal
  on Scientific Computing, 43 (2021), pp.~A1134--A1162,
  \url{https://doi.org/10.1137/20M1349965},
  \url{https://doi.org/10.1137/20M1349965},
  \url{https://arxiv.org/abs/https://doi.org/10.1137/20M1349965}.

\bibitem{popov2020explicit}
{\sc A.~A. Popov and A.~Sandu}, {\em An explicit probabilistic derivation of
  inflation in a scalar ensemble {K}alman filter for finite step, finite
  ensemble convergence}, 2020, \url{https://arxiv.org/abs/2003.13162}.

\bibitem{popov2022multifidelitychapter}
{\sc A.~A. Popov and A.~Sandu}, {\em Multifidelity data assimilation for
  physical systems}, in Data Assimilation for Atmospheric, Oceanic and
  Hydrologic Applications (Vol. IV), Springer, 2022, pp.~43--67.

\bibitem{Popov2021b}
{\sc A.~A. Popov and A.~Sandu}, {\em Multifidelity ensemble {Kalman} filtering
  using surrogate models defined by theory-guided autoencoders}, Frontiers in
  Applied Mathematics and Statistics, accepted (2022).

\bibitem{reich2013nonparametric}
{\sc S.~Reich}, {\em A nonparametric ensemble transform method for bayesian
  inference}, SIAM Journal on Scientific Computing, 35 (2013),
  pp.~A2013--A2024.

\bibitem{reich2015probabilistic}
{\sc S.~Reich and C.~Cotter}, {\em {Probabilistic forecasting and Bayesian data
  assimilation}}, Cambridge University Press, 2015.

\bibitem{otp}
{\sc S.~Roberts, A.~A. Popov, and A.~Sandu}, {\em {ODE} test problems: a
  {MATLAB} suite of initial value problems}, 2019,
  \url{https://arxiv.org/abs/1901.04098}.

\bibitem{sendera2020supermodeling}
{\sc M.~Sendera, G.~S. Duane, and W.~Dzwinel}, {\em Supermodeling: the next
  level of abstraction in the use of data assimilation}, in International
  Conference on Computational Science, Springer, 2020, pp.~133--147.

\bibitem{sirovich1987turbulence1}
{\sc L.~Sirovich}, {\em Turbulence and the dynamics of coherent structures.
  {I.} coherent structures}, Quarterly of applied mathematics, 45 (1987),
  pp.~561--571.

\bibitem{subrahmanya2021ensemble}
{\sc A.~N. Subrahmanya, A.~A. Popov, and A.~Sandu}, {\em Ensemble variational
  fokker-planck methods for data assimilation}, 2021,
  \url{https://doi.org/10.48550/ARXIV.2111.13926},
  \url{https://arxiv.org/abs/2111.13926}.

\bibitem{whitaker2002ensemble}
{\sc J.~S. Whitaker and T.~M. Hamill}, {\em Ensemble data assimilation without
  perturbed observations}, Monthly weather review, 130 (2002), pp.~1913--1924.

\bibitem{xue2014multimodel}
{\sc L.~Xue and D.~Zhang}, {\em A multimodel data assimilation framework via
  the ensemble kalman filter}, Water Resources Research, 50 (2014),
  pp.~4197--4219.

\end{thebibliography}
